\documentclass{article}
\usepackage[utf8]{inputenc}
\usepackage{a4wide}
\usepackage[UKenglish]{babel}
\usepackage{graphicx}
\usepackage{subfigure}
	\graphicspath{{./}{figures/}}
\usepackage[colorlinks=true,linkcolor=blue,citecolor=red]{hyperref}
\usepackage{amssymb,empheq,bbold}
\usepackage[inline]{enumitem} 
\usepackage{geometry}
\usepackage{caption}
\usepackage{amsthm}
\usepackage{bbm}

\newcommand \commentout[1] {}

\newcommand{\abs}[1]{\left\lvert#1\right\rvert}

\newcommand{\R}{\mathbb{R}}

\newtheorem{thm}{Theorem}
\newtheorem{lemma}{Lemma}
\newtheorem*{remark}{Remark}
\newtheorem{assumption}{Assumption}
\title{Asymptotic and stability analysis of kinetic models for opinion formation on networks: an Allen-Cahn approach}
\author{Martin Burger \thanks{{Helmholtz Imaging, Deutsches Elektronen-Synchrotron DESY, Notkestr. 85, 22607
Hamburg, Germany.} \\ {  Fachbereich Mathematik, Universit\"at Hamburg, Bundesstrasse
55, Hamburg, 20146, Germany}, \texttt{martin.burger@desy.de}}
, Nadia Loy\thanks{{Dipartimento di Scienze Matematica ``G.L. Lagrange'', Politecnico di Torino, Corso Duca degli Abruzzi, 24, 10129 Torino, Italy,} \texttt{nadia.loy@polito.it}}, Alex Rossi}
\date{\today}

\begin{document}

\maketitle

\begin{abstract}
We present the analysis of the stationary equilibria and their stability in case of an opinion formation process in presence of binary opposite opinions evolving according to majority-like rules on social networks. The starting point is a kinetic Boltzmann-type model derived from microscopic interactions rules for the opinion exchange among individuals holding a certain degree of connectivity. The key idea is to derive from the kinetic model an Allen-Cahn type equation for the fraction of individuals holding one of the two opinions. The latter can be studied by means of a linear stability analysis and by exploiting integral operator analysis. While this is true for ternary interactions, for binary interactions the derived equation of interest is a linear scattering equation, that can be studied by means of General Relative Entropy tools and integral operators.
\end{abstract}

\medskip


\medskip


\section{Introduction}
 Opinion dynamics has raised a relevant interest in recent years, and has been extensively studied by means of mathematical models~\cite{nguyen2020SciRep,choi2019NJP,Noonan2021PRE,peng2022,lambiotte2007PRE}.
A popular approach is to model the opinion as a binary value representing two different opposite and competing opinions. The opinion formation and evolution process is then described as a consequence of interactions between individuals. In this sense, the possible connections among individuals are of utmost interest, as people may be affected by those who are connected to them in such a way to match the opinion of the majority. This is the reason why a prominent interest is represented by opinion dynamics on complex social networks, where each individual holds a binary opinion that may change in consequence of majority-like rules.  In this framework, both the opinion dynamics and the network may lead the system to an equilibrium in the final steady state, where a certain opinion dominates the system
or coexists with the other one. 
A large number of theoretical investigations rely on the construction of graph models for complex networks~\cite{newman2002PNAS,watts1998NATURE}. 
In this direction there are works relying on physical approaches or probabilistic ones, in which from majority-like rules on graphs, equations for aggregate quantities are approximated or proposed by means of mainly euristic approaches \cite{nguyen2020SciRep,choi2019NJP,Noonan2021PRE,peng2022,lambiotte2007PRE}. In these works, the model relies on the description of the evolution of the probability of having one of the two opposite opinions in a node of a graph that can be characterized by its degree of connectivity (that in a finite graph is given by the number of connections). Such a probability is assumed to evolve according a markovian process described by a master equation. From the latter, as done for example in~\cite{Noonan2021PRE}, a Fokker-Planck equation and the corresponding stochastic differential equation are derived and, assuming a small negligible diffusion, the model becomes a system of ODEs whose linear stability is analysed. Other works, such as~\cite{peng2022}, rely on linear stability analysis of an equilibrium determined using the Perron-Frobenius theory and numerics, or fully rely on numerical simulations~\cite{nguyen2020SciRep}. In these works, the network is often modelled as a finite graph, but also as an hypergraph or as a multiplex network, in order to describe hierarchies in the connections.

Concerning the description of the network, other works rely on the characterisation of the \textit{statistical} structure of social networks~\cite{albert2002RMP,barabasi1999SCIENCE,barabasi1999PHYSA}. The reason for a statistical approach is due to the fact that the number of nodes and links of a social network is so large that a detailed description by means of classical graphs can be hardly manageable.
In the present work, we consider mesoscopic models which incorporate a statistical description of the connectivity of the individuals. In this description, the individuals are characterised by a discrete binary opinion ($\pm 1$) and by a positive \textit{degree of connectivity}.  The social network is described statistically by means of a probability density function for the degree of connectivity, as done for example in~~\cite{albi2017KRM,toscani2018PRE, calzola2023, LnMrTa}. Microscopic stochastic processes describing the interactions among individuals and leading to an exchange of the opinion given a certain distribution of the connectivity can be stated and allow to derive a corresponding \textit{kinetic model}~\cite{LnMrTa}. Such kinetic models describe the time evolution of a distribution function of the opinion and of the degree of connectivity. Kinetic models also allow to derive evolution equations for collective averaged quantities, such as the average opinion on the network. 

Specifically, we consider the Boltzmann-type kinetic model studied in~\cite{LnMrTa}. This kinetic model implements a microscopic interaction that is a \textit{Two-against-one} mechanism  (that is a kind of majority rule, in the same spirit as~\cite{slanina2003EPJB}), that is ruled by a specific \textit{interaction kernel},  while the social network is described statistically by means of a probability density function for the degree of connectivity. In \cite{LnMrTa} the authors consider one specific interaction kernel that is proportional to the connectivity of the influencing agents, and manage to derive, from the kinetic Boltzmann-type description, an evolution equation of the average opinion on the network. The latter allows to determine the stationary asymptotic equilibria and to characterise them. Here, we consider a generic interaction kernel. The aim is to determine the asymptotic states and their stability for each interaction kernel. This is not banal, specially far from a linear stability regime, for a Boltzmann-type equation. Then, our leading idea is that determining the stationary sates and their stability through the analysis of an Allen-Cahn equation is easier, even for a generic interaction kernel and at least in a linear regime, than for the Boltzmann-type equation. As a consequence, we formally derive from the kinetic model an Allen-Cahn equation for the evolution of the opinion distribution over the network and study its stationary equilibria and their stability.  %We also apply the same kind of analysis to the \textit{Ochrombel simplification} model~\cite{LnMrTa}, that leads to a linear scattering equation. 

In Section 2, we revise the kinetic model proposed in~\cite{LnMrTa} implementing the \textit{Two-against-one} intereaction, and the related essential statistical properties of the network and of the opinion on the network. Then, we derive an appropriate Allen-Cahn equation and perform a linear stability analysis. In Section 3 we consider the \textit{Ochrombel simplified model} and perform a similar analysis. While in Section 2 and 3 the network is considered to be fixed and constant in time, in Section 4, both the Two-against-one and the Ochrombel simplified model are considered on a co-evolving network, and similar results on the stationary equilibria are presented. In Section 5 we show some numerical tests, while in Section 6 we draw some conclusions.

\section{Opinion formation models and their Allen-Cahn description}
In this section, firstly we briefly revise the kinetic model for opinion formation as proposed in~\cite{LnMrTa}, then we derive from the kinetic description an Allen-Cahn equation for the opinion dynamics, and, eventually, we perform the linear stability analysis of the latter model.
\subsection{Preliminaries on the Boltzmann-type description}
Let us consider a large number of agents interacting through a social network and exchanging their opinion. The microscopic state of a generic user is defined by their opinion-connectivity pair $(w,\,c)$, where $w\in\{-1,\,1\}$ is a discrete binary variable and $c\in\R_+$ is a continuous positive one (being $c \equiv 0$ a degenerate case). The values $w=\pm 1$ conventionally denote two opposite opinions in the same spirit as the Sznajd model~\cite{sznajd-weron2000IJMP,sznajd-weron2021PHYSA}. The quantity $c$ is assumed to be a representative measure of the degree of connectivity of a given individual, namely of the number of users who may be exposed to the opinion expressed by that individual. We consider this variable to be continuous in accordance with the reference literature on the connectivity distribution of social networks, see e.g.,~\cite{barabasi1999SCIENCE,barabasi1999PHYSA,clauset2009SIREV,newman2002PNAS,watts1998NATURE}.

Let $f=f(w,c,t)$ be the kinetic distribution function of the pair $(w,\,c)$ at time $t\geq 0$. Owing to the discreteness of $w$, we may represent it as
\begin{equation}
	f(w,c,t)=p(c,t)\delta(w-1)+q(c,t)\delta(w+1),
	\label{eq:f}
\end{equation}
where $\delta(w-w_0)$ denotes the Dirac delta distribution centered at $w=w_0$ and $p,\,q\geq 0$ are coefficients which depend in general on $c$ and $t$ and represent the portion of agents with opinion $w=\pm1$, respectively, having connectivity $c$ at time $t$. Moreover, we assume $p(\cdot,t),\,q(\cdot,t)\in L^1(\R_+)$ for all $t\geq 0$. Considering a constant-in-time number of users and connections of the social network, we may impose the normalisation condition
\begin{equation}
	\int_{\R_+}\int_{\{-1,\,1\}}f(w,c,t)\,dw\,dc=1 \qquad \forall\,t\geq 0,
	\label{eq:f.normalisation}
\end{equation}
and consequently think of $f$ as the probability density of the pair $(w,c)$. The probability density distribution of the opinion is then given by the marginal
$$ h(w,t):=\int_{\R_+}f(w,c,t)\,dc=\hat{p}(t)\delta(w-1)+\hat{q}(t)\delta(w+1), $$
where we have set
$$ \hat{p}(t):=\int_{\R_+}p(c,t)\,dc, \qquad \hat{q}(t):=\int_{\R_+}q(c,t)\,dc. $$
Notice that $\hat{p}(t)$, $\hat{q}(t)$ are the probabilities that an individual expresses the opinion $w=1$ or $w=-1$, respectively, at time $t$. Consistently with~\eqref{eq:f.normalisation}, it results
\begin{equation}
	\hat{p}(t)+\hat{q}(t)=1, \qquad \forall\,t\geq 0.
	\label{eq:phat+qhat}
\end{equation}
The connectivity density of the social network is instead given by the marginal
$$ g(c,t):=\int_{\{-1,\,1\}}f(w,c,t)\,dw=p(c,t)+q(c,t). $$
In particular, $g$ will be an assigned probability density function satisfying
\begin{equation}\label{def:g}
    \int_{\R_+} g(c,t) \, dc =1, \qquad \forall t \ge 0.
\end{equation}
We also define
\begin{equation}
    m_C(t):=\int_{\R_+} c~g(c,t) dc
\end{equation}
being the average connectivity, that we assume to be finite. 

If we assume that the number of followers of a generic individual possibly varies in time more slowly than their opinion, and we actually do not specify a dynamics for the connectivity, then the marginal distribution $g$ may be well considered constant in time, i.e. $g(c,t)=g(c)$ for all $t\geq 0$. Consequently, the sum $p(c,t)+q(c,t)$ is constant in $t$ although the single terms $p(c,t)$, $q(c,t)$ may be not. Of course, when $g$ does not depend on time, then $m_C$ is a constant.

However, a priori, the probability density function $f$ may be factorized as follows
\begin{equation}\label{def:fact}
 f(w,c,t)=g(c,t)h_c(w,t) \quad \forall\,t\geq 0,   
\end{equation}
where $h_c(w,t)$ is the conditional probability density of the opinions of the agents given the degree of connectivity $c$. Assuming statistical independence of $w$ and $c$ corresponds to
\begin{equation}\label{def:indep}
    h_c(w,t)=h(w,t),
\end{equation}
i.e. $f$ is totally factorized with respect to the dependence on $w$ and $c$.

Following standard procedures, in~\cite{LnMrTa} the authors derive formally a continuous-in-time statistical description in terms of kinetic Boltzmann-type ``collisional'' equations for the distribution function $f$. Such a kinetic equation, in weak form, using an arbitrary \textit{observable quantity} that is a test function of the microscopic state $(w,c)$, i.e. $\phi=\phi(w,c):\{-1,\,1\}\times\R_+\to\R$, reads
\begin{multline}
	\frac{d}{dt}\int_{\R_+}\int_{\{-1,\,1\}}\phi(w,c)f(w,c,t)\,dw\,dc \\
	=\frac{1}{3}\int_{\R_+^3}\int_{\{-1,\,1\}^3}\bar{B}(w,c,w_\ast,c_\ast,w_{\ast\ast},c_{\ast\ast})\bigl(\phi(w',c)-\phi(w,c)\bigr) \\
		\times f(w,c,t)f(w_\ast,c_\ast,t)f(w_{\ast\ast},c_{\ast\ast},t)\,dw\,dw_\ast\,dw_{\ast\ast}\,dc\,dc_\ast\,dc_{\ast\ast},
	\label{eq:2-vs-1.Boltz}
\end{multline}
where $w'=w_{\ast\ast}$ and
where the \textit{collision kernel} $\bar{B}$ is
\begin{equation}\label{def.B}\bar{B}(w,c,w_\ast,c_\ast,w_{\ast\ast},c_{\ast\ast}):=\eta\delta(w_\ast-w_{\ast\ast})\omega(c,c_\ast,c_{\ast\ast}), 
\end{equation}
being $\eta>0$ a proportionality constant (the frequency of the interactions). % and $\chi(\cdot)$ denotes the characteristic function of the event indicated in parenthesis. 
The latter implements a majority rule, which means that the opinion changes according to the \textit{two-against-one} model, which assumes \textit{ternary} interactions in which the first individual ($w$) takes the opinion ($w'=w_{\ast\ast}$) of the second two ones ($w_\ast, w_{\ast\ast}$) if they have the same opinion (if $w_{\ast\ast}=w_\ast$); otherwise, the interaction takes place but leas to no exchange. The probability that the first individual is reached by the common opinion of the second two individuals and gets convinced by them is proportional to the generic \textit{interaction kernel} $\omega(c,c_\ast,c_{\ast\ast})$, the gives a measure of the interaction rate depending on the degree of connectivity of the individuals.
In this model, the connectivity does not change and, therefore, it is assumed that $g(c)$ is constant in time and assigned.

In~\cite{LnMrTa}, the authors choose $\omega(c,c_\ast,c_{\ast\ast})=c_\ast c_{\ast\ast}$, i.e. the probability that the first agent is reached by the opinion of the second two ones is proportional to the connectivity of the second two agents. In this case, it is possible to compute the exact evolution equation of the macroscopic average opinion 
$$
M_w:=\int_{\mathbb{R}_+}\int_{\lbrace -1,1\rbrace} f(w,c,t) \, w dc dw.
$$ 
Bu studying the evolution equation for $M_w$, in~\cite{LnMrTa} the authors show that for the specific choice $\omega(c,c_\ast,c_{\ast\ast})=c_\ast c_{\ast\ast}$, $M_w=\pm 1$ are stable stationary states, i.e. there is polarization of the social network. This corresponds to $\hat{p}=0$ or $\hat{p}=1$, respectively. Else, the vanishing average opinion $M_w=0$, that represents coexistence of the two different opinions and corresponds to $\hat{p}=\frac{1}{2}$, is unstable. 
It is also shown that if there is statistical independence, i.e. $h_c(w,t)=h(w)$ does not depend on $c$, then the average opinion has the same sign as the initial one, i.e. there is no switch in the opinion polarization. Specifically,
it is shown that
\begin{equation}
    \frac{d\hat{p}}{dt}=\dfrac{\eta}{3} m_C^2 \hat{p}(1-\hat{p})(2\hat{p}-1),
\end{equation}
which admits $\hat{p}=0,1$ as stable states and $\hat{p}=\frac{1}{2}$ as unstable ones.
\\

\subsection{Allen-Cahn description}\label{sec:allen_cahn}
We now want to derive from the Boltzmann-type description~\eqref{eq:2-vs-1.Boltz} an Allen-Cahn equation.
Let us consider the weighted distributions of the connectivity, i.e.
\begin{equation}\label{def:tildepq}
    \tilde{p}(c,t):=\dfrac{p(c,t)}{g(c,t)}, \quad \tilde{q}(c,t):=\dfrac{q(c,t)}{g(c,t)}.
\end{equation}
Specifically, we remind that, in~\eqref{eq:2-vs-1.Boltz}, $g(c)$ is time independent.
We remark that now
\begin{equation}
    \tilde{p}(c,t)+\tilde{q}(c,t)=1,
\end{equation}
i.e. $\tilde{p},\tilde{q}$ still feature their dependence on $c$, but their sum is constantly equal to one, and they may be thus seen as probabilities.
Considering the Boltzmann-equation \eqref{eq:2-vs-1.Boltz} with $\bar{B}$ defined as in \eqref{def.B} and the test function of the form $\phi(w,c)=\varphi(w)\psi(c)$, we get the $c$-strong form of the Boltzmann-equation
\begin{multline}
	\frac{d}{dt}\int_{\{-1,\,1\}}\varphi(w)f(w,c,t)\,dw 
	=\frac{\eta}{3}\int_{\R_+^2}\int_{\{-1,\,1\}^3}\omega(c,c_\ast,c_{\ast\ast})\bigl(\varphi(w')-\varphi(w)\bigr) \\
		\times f(w,c,t)f(w_{\ast\ast},c_\ast,t)f(w_{\ast\ast},c_{\ast\ast},t)\,dw\,dw_\ast\,dw_{\ast\ast}\,dc_\ast\,dc_{\ast\ast}, \qquad w'=w_{\ast\ast}.
	\label{eq:Boltz_c_strong}
\end{multline}
If we now set $\varphi(w)=\dfrac{w+1}{2}$ in \eqref{eq:Boltz_c_strong}, divide \eqref{eq:Boltz_c_strong} by $g(c)$ and augment the integrand of the right-hand side of \eqref{eq:Boltz_c_strong} by $\dfrac{g(c_\ast)g(c_{\ast\ast})}{g(c_\ast)g(c_{\ast\ast})}$, then we obtain the equation for the weighted distribution of the connectivity $\tilde{p}$  that is
\begin{equation}\label{eq:Allen_Cahn}
\partial_t \tilde{p}(c,t)=\mathcal{D}(\tilde{p})+\Gamma(c)\left[\tilde{p}(1-\tilde{p})(2\tilde{p}-1)\right],    
\end{equation}
where
\begin{multline}\label{def.D}
  \mathcal{D}(\tilde{p}):=\int_{\R_+^2}\omega(c,c_\ast,c_{\ast\ast})\left[\left(\tilde{p}(c_\ast,t)\tilde{p}(c_{\ast\ast},t)-\tilde{p}^2(c,t)\right)\left(1-\tilde{p}(c,t)\right)\right.\\
  \left.-\tilde{p}(c,t)\left((1-\tilde{p}(c_\ast,t))(1-\tilde{p}(c_{\ast\ast},t))-(1-\tilde{p}(c,t))^2\right) \right]g(c_\ast)g(c_{\ast\ast}) \, dc_\ast \, dc_{\ast\ast},  
\end{multline}
and
%being $\tilde{\omega}(c,c_\ast,c_{\ast\ast}):=\omega(c,c_\ast,c_{\ast\ast}) g(c_\ast)g(c_{\ast\ast})$
%and
\begin{equation}\label{def.Gamma}
\Gamma(c):=\int_{\R_+^2}\omega(c,c_\ast,c_{\ast\ast}) g(c_\ast)g(c_{\ast\ast})\, dc_\ast \, dc_{\ast\ast}.   
\end{equation}
It can be easily observed that $\mathcal{D}$ annihilates constants, and hence the zeros of the reaction term, that are
\begin{equation}\label{def.stationary_states}
 \tilde{p}_0=0, \qquad \tilde{p}_1=1, \qquad \tilde{p}_2=\dfrac{1}{2},   
\end{equation}
are the constant stationary solutions of~\eqref{eq:Allen_Cahn}-\eqref{def.D}-\eqref{def.Gamma}. \\
We remark that $\tilde{p}_0=0$ corresponds to a stationary state in which all the agents have opinion $-1$, while $\tilde{p}_1=1$ corresponds to a stationary state in which all the agents have opinion $+1$, i.e. $\tilde{p}_{0,1}$ depict a polarization scenario. The third stationary state $\tilde{p}_2=\dfrac{1}{2}$, implying $\tilde{q}_2=\dfrac{1}{2}$, depicts a 
 non-polarized scenario in which half of the population has opinion $+1$ and the other half $-1$, thus implying coexistence of the two opposite opinions. According to this notation we have that
 \[
 h_c(w,t)= \tilde{p}(c,t)\delta(w-1)+\tilde{q}(c,t)\delta(w+1),
 \]
and then we can link $\tilde p$ to the average opinion $m_w(c,t)$ for a given connectivity $c$, that is defined as
    $$
    m_w(c,t) \coloneqq \int_{\lbrace{-1, 1\rbrace}}w h_c(w, t)\,dw = \tilde{p}(c,t)-\tilde{q}(c,t)= g(c,t)\left(p(c,t)- q(c,t)\right),
    $$
    and remark that the total average opinion over the network is
    \begin{equation}\label{def:Mw}
    M_w(t)= \int_{\R_+} m_w(c,t) g(c,t) \, dc = \int_{\R_+} \left(p(c,t)-q(c,t)\right) \, dc.
    \end{equation}
%Going back to the previous description, as we have defined~\eqref{def:tildepq}, 
Then, if $\tilde{p}^{\infty}$ is a constant stationary state, the actual stationary state is $c$-dependent as given by $p^\infty (c)=\tilde p^\infty g^\infty(c)$ (whose existence is to be discussed later and that is simply $g(c)$ when the connectivity is fixed). 
    Therefore:
    $$m_w^\infty(c) =  (\tilde p^\infty - \tilde q^\infty) = (2\tilde p^\infty - 1),$$
    which means that the normalization function $g$ and the constant $\tilde p_\infty$ determine the average opinion.
    The latter, together with~\eqref{def:Mw}, also implies that
    \begin{equation}\label{eq:Mw-ptilde}
    M_w^\infty=(2\tilde p^\infty - 1).
    \end{equation}
    %Integrating we obtain $p^{\infty}=\tilde{p}^{\infty}=\int_{\R_+} \dfrac{p(c,0)}{g(c)} \, dc$.
    %Therefore
    %$$ m_W^{\infty}=2\tilde{p}^{\infty}-1 $$
    %Then, when you obtain a certain $\tilde{p}^{\infty}$, you get the corresponding average opinion.
%The specific case considered in~\cite{LnMrTa} may be translated here to the generic interaction kernel $\omega$, by looking at the average opinion $M_w$~\eqref{eq:Mw-ptilde}. 
We remark that~\eqref{eq:Mw-ptilde}, that holds for every $\omega$, implies that the polarized scenarios described by $M_w^\infty=\pm 1$ correspond to $\tilde p^\infty=\tilde{p}_{0,1}$ (the polarized configurations). On the other hand, the coexistence of the two different opinions, described by $M^\infty_w=0$, corresponds to $\tilde p^\infty=\tilde{p}_2$. In the specific choice considered in~\cite{LnMrTa}, it is shown that $\tilde{p}_{0,1}$ are stable, while the coexistence configuration $\tilde{p}_2$ is unstable.  Our aim, now, is to analyse the stability of the stationary states $\tilde{p}_i, \, i=0,1,2$ for a generic $\omega$, in order to see if the non-polarized scenario is always unstable and if the polarized scenarios are always stable.
\subsection{Linear stability analysis of the Allen-Cahn model}\label{sec:lin_stab}
From now on we shall drop the $\tilde{ }$ in the notation, bearing, then, in mind the fact that $p(c,t)$ is the fraction of agents having connectivity $c$ at time $t$ and it may be constant.
Let us now consider a perturbation of the stationary states, i.e.
\[
p_i+ k_i(c,t), \qquad i=0,1,2.
\]
Plugging the latter in \eqref{eq:Allen_Cahn}, it can be easily verified that
\[
\mathcal{D}(p_i+k_i) =\mathcal{D}(p_i)+\mathcal{D}'(p_i)k_i +\mathcal{O}(k_i^2),
\]
where
\begin{multline}\label{def.D_prime} \mathcal{D}'(p)k=\int_{\R_+^2}\tilde \omega(c,c_\ast,c_{\ast\ast})\left[\big(p(c_\ast,t)k(c_{\ast\ast},t) +p(c_{\ast\ast},t)k(c_\ast,t) -2p(c,t)k(c,t)\big)(1-p(c,t))\right.\\ \left.-k(c,t)\big(p(c_\ast,t)p(c_{\ast\ast},t)-p(c,t)^2\big)-k(c,t)\big((1-p(c_\ast,t))(1-p(c_{\ast\ast},t))-(1-p(c,t))^2\big)\right.\\
\left.+p(c,t)\big(k(c_\ast,t)(1-p(c_{\ast\ast},t))+(1-p(c_\ast,t))k(c_{\ast\ast},t)-2(1-p(c,t))k(c,t)\big)\right] \, dc_\ast dc_{\ast\ast},
\end{multline}
and, therefore
\begin{align*}
    p_0=0:& \qquad \partial_t k_0=\mathcal{D}'(p_0)k_0-\Gamma k_0;\\
    p_1=1:& \qquad \partial_t k_1=\mathcal{D}'(p_1)k_1-\Gamma k_1;\\
    p_2=\dfrac{1}{2}:& \qquad \partial_t k_2=\mathcal{D}'(p_2)k_2+\dfrac{\Gamma}{2} k_2.
\end{align*}
We remark that for $i=0,1$ $\mathcal{D}(p_i)=\mathcal{D}'(p_i)=0$: this implies stability of the fully polarized configurations. On the other hand, for the stationary state $p_2$ we get
\begin{equation}
    \partial_t k_2=\int_{\R_+^2}\dfrac{\omega(c,c_\ast,c_{\ast\ast})}{2}g(c_\ast)g(c_{\ast\ast})\big(k_2(c_\ast,t)+k_2(c_{\ast\ast},t)-2k_2(c,t)\big) \, dc_\ast \, dc_{\ast\ast}+\dfrac{\Gamma(c)}{2}k_2.
\end{equation}
If we now consider $k_2= e^{\lambda t} k(c)$, the latter becomes
\begin{equation}\label{def.eigenvalue_pb}
    \lambda k(c)=\int_{\R_+^2}\dfrac{\omega(c,c_\ast,c_{\ast\ast})}{2}g(c_\ast)g(c_{\ast\ast})[k(c_\ast)+k(c_{\ast\ast})]\, dc_\ast \, dc_{\ast\ast}-\Gamma(c) k(c)+\dfrac{\Gamma(c)}{2}k(c) .
\end{equation}
Hence, we consider the following
\begin{assumption}\label{ass.1}
The kernel $\omega > 0$ is symmetric with respect to the connectivity of the influencing individuals, i.e. $$\omega(c,c_\ast,c_{\ast\ast})=\omega(c,c_{\ast\ast},c_\ast).$$
\end{assumption} 
The latter means that there is, reasonably, symmetry with respect to the two influencing agents.
As a consequence~\eqref{def.eigenvalue_pb} becomes
\begin{equation}\label{def.eigenvalue_pb2}
    \lambda k(c)=\int_{\R_+^2}\omega(c,c_\ast,c_{\ast\ast})g(c_\ast)g(c_{\ast\ast})k(c_\ast)\, dc_\ast \, dc_{\ast\ast}-\dfrac{\Gamma(c)}{2}k(c),
\end{equation} 
that reduces to
\begin{equation}\label{eq:eigenvalue_gen}
    k(c)=\dfrac{\int_{\R_+^2} \omega(c,c_\ast,c_{\ast\ast}) g(c_\ast)g(c_{\ast\ast}) k(c_\ast)\, dc_\ast \, dc_{\ast\ast}}{\lambda+\Gamma(c)/2}.
\end{equation}
 
As our guess is that the coexistence configuration $\tilde{p}_2$ is always unstable, we want to show the existence of a positive eigenvalue $\lambda$ by studying~\eqref{eq:eigenvalue_gen}. We shall do this by exploiting an indirect argument. First of all, we introduce the measure space 
$$
L^1(\R_+,\mu), \qquad d\mu(c)= g(c) dc,
$$
that is a finite space measure as $g$ is a probability density function on $\R_+$ (see~\eqref{def:g}) and we denote it $L^1_\mu$. We consider the space $L^1(\R_+,\mu)$ to be endowed with its natural $L^1$ norm 
\begin{equation}
   \|\phi\|_{L^1(\R_+,\mu) } \coloneqq \int_{\R_+} \phi(c) g(c) \, dc.
\end{equation}
We remark that, as $p$ defined in~\eqref{eq:f} is in $L^1(\R_+)$, then $\tilde{p}$ defined in~\eqref{def:tildepq} naturally belongs to $L^1(\R_+,\mu)$, and the same holds for the perturbation $k$, that is actually a perturbation of $\tilde{p}$.
Specifically we have that
\[
\|\tilde{p}\|_{L^1(\R_+,\mu) }=\|p\|_{L^1(\R_+)}.
\]
We also make the further, reasonable, assumption on $\omega$, namely
\begin{assumption}\label{ass.2}
    The kernel $\omega > 0$ and the distribution of the connectivities $g$ are defined in such a way to satisfy 
    \begin{equation}\label{propr.1}
    \omega \in L^\infty(\R_+\times \R_+\times\R_+, \mu \times \mu \times \mu). 
\end{equation}
\end{assumption}
We now introduce
\begin{equation}\label{def_G}
    G(c,c_\ast):=\int_{\R_+} \omega(c,c_\ast,c_{\ast\ast}),  \, d\mu(c_{\ast\ast}) 
\end{equation}
and we remark that~\eqref{propr.1} implies that
\begin{equation}\label{G_Linfty}
    G\in L^\infty(\R_+\times \R_+, \mu\times \mu),
\end{equation}
and also
\begin{equation}\label{propr.2}
    \Gamma \in L^\infty(\R_+, \mu).
\end{equation}
Moreover, as $\omega > 0$, both $G$ and $\Gamma$ are positive.

\noindent Next, we introduce the operator
%\begin{equation}
    \begin{align}
      A_\lambda : L^1(\R_+,\mu) &\longmapsto   L^1(\R_+,\mu)\\
       k &\longmapsto A_\lambda k(c) := \int_{\R_+} G_\lambda (c,c_\ast) k(c_\ast)  \, d\mu(c_\ast), 
    \end{align}
%\end{equation}
where
\begin{equation}\label{def_Gl}
  G_\lambda(c,c_\ast)= \dfrac{G(c,c_\ast)}{\lambda+\Gamma(c)/2}.  
\end{equation}
We remark that, as $\Gamma>0$, then for each $\lambda$ we have that $\lambda +\Gamma/2>\lambda$, and therefore $G_\lambda <\dfrac{G}{\lambda}$, so that
\begin{equation}\label{propr_Gl_infty}
    G_\lambda \in L^\infty(\R_+\times \R_+, \mu\times \mu),
\end{equation}
as~\eqref{G_Linfty} holds. Then, for the Holder inequality,  $A_\lambda k(c) \in L^1(\R_+,\mu)$ and the operator $A_\lambda$ is well defined. Specifically, because of~\eqref{propr_Gl_infty}, $A_\lambda$ is a bounded operator and for the embedding of finite measure spaces, then we also have that
\begin{equation}\label{propr_Gl}
    G_\lambda \in L^2(\R_+\times\R_+, \mu\times \mu),
\end{equation}
then  $A_\lambda$ is compact~\cite{Krasno}.

Now, we study for fixed $\lambda$ the eigenvalue problem
$$  A_\lambda k(c)= \delta k(c)  $$
and we want to show that there exists $\lambda>0$ for the eigenvalue $\delta=1$ with a positive eigenfunction. This corresponds to prove that there exists $\lambda>0$ for \eqref{eq:eigenvalue_gen}.
We call $\mathcal{K} \subset L^1(\R_+,\mu)$ the cone of nonnegative functions, that is reproducing.  
For $\lambda >0$ 
%$$
%\lambda > -\max_{\R_+} \dfrac{\Gamma(c)}{2}
%$$
$A_\lambda$ is a strongly positive operator on $\mathcal{K}$.
Moreover, as a consequence of~\eqref{propr_Gl_infty}, that is granted by Assumption \ref{ass.2}, %~\eqref{propr.1}-\eqref{propr.2},for $\lambda>0$, 
we have that
\begin{equation}
    \dfrac{\Gamma(c)}{\lambda +\dfrac{\Gamma(c)}{2}} \in L^\infty(\R_+, \mu).
\end{equation}
Therefore, we can apply the Krein-Rutman theorem and, then, for each $\lambda>0$, we have that there exists $\delta=\delta(\lambda)>0$, that is also the spectral radius of $A_\lambda$, associated to a positive eigenfunction that is the only positive eigenfunction. 
We can remark that for $k \equiv 1$ , $\delta(\lambda)=1$ is an eigenvalue-eigenfunction couple if
\begin{equation}\label{eq:lambda_norm}
    \int_{\R_+} \dfrac{\Gamma(c)}{\lambda +\dfrac{\Gamma(c)}{2}} g(c) \, dc=1.
\end{equation}
The latter is not verified if $\lambda <0$. In fact, let us suppose that there exists $\lambda <0$ satisfying~\eqref{eq:lambda_norm}. Then $\lambda+\Gamma(c)/2 < \Gamma (c)/2$. Therefore $\dfrac{\Gamma(c)/2 }{\lambda+\frac{\Gamma(c)}{2}}>1$, and then $\int_{\R_+} \dfrac{\frac{\Gamma(c)}{2} g(c)}{\lambda+\frac{\Gamma(c)}{2}} \, dc > \int_{\R_+} g(c) \, dc =1$, which is a contradiction. 
We also remark that, if $\omega$ does not depend on $c$, we find
    \[
     1=\dfrac{\Gamma}{\lambda+\Gamma/2},
    \]
   implying
    \[
     \lambda=\dfrac{\Gamma}{2}.
    \]
Specifically, when $\omega(c,c_\ast,c_{\ast\ast})=c_\ast c_{\ast\ast}$ as in~\cite{LnMrTa},  then $\Gamma=m_C^2$, so that 
\[
\lambda=m_C^2/2>0.
\]
%and 
%\begin{equation}
%    \lambda= \int_{\R_+} \dfrac{\Gamma(c)}{2}  g(c) \, dc
%\end{equation}
%$\delta(\lambda)=1$ is an eigenvalue. 
In conclusion, we have found a positive $\lambda$ and, then, $p_2=\frac{1}{2}$ is unstable.

We can summarize these observations in one theorem.
\begin{thm} \label{thm_pos}
    Let $\omega$ satisfy Assumptions~\ref{ass.1}-\ref{ass.2}.
    %Let the kernel $G_\lambda$ be defined by \eqref{def_Gl}. 
    Then the operator $\mathcal D'\left(\frac 12\right) + \frac \Gamma 2 Id$ has a positive eigenvalue and this means that the stationary state $\tilde{p}_2 = \frac 12$ of equation \eqref{eq:Allen_Cahn} is linearly unstable.
\end{thm}
This allows to conclude that the stationary kinetic distribution of the opinions is the Dirac delta centered in $1$ or $-1$, i.e. $h^\infty(w)=\delta(w\pm 1)$.

\begin{remark}
We remark that Assumption~\ref{ass.2} that concerns the regularity of $\omega$ in the finite measure space $L_1(\R_+^3,\mu^3)$, that is defined by $g$, depends on the interplay of $\omega$ and $g$.
Specifically, we remark that kernels in the form $\omega(c,c_\ast,c_{\ast\ast})=c c_\ast c_{\ast\ast}$ or $\omega(c,c_\ast,c_{\ast\ast})= c_\ast c_{\ast\ast}$ are positive as $c, c_\ast,c_{\ast\ast} >0$. 
%In order to consider this kind of interaction kernel to be in $L^\infty$ we can consider the following assumption
%\begin{assumption}\label{ass.0}
%    The moments of all orders of $g$ are finite and
%    \[
%    \sup_{p>0} || c^p g||_{L^1(\R_+)} < \infty.
%    \]
%\end{assumption}
%In fact, if $g$ satisfies Assumption~\ref{ass.0} and we are in a finite measure space, then, for example in the case $\omega(c,c_\ast,c_{\ast\ast})=c_\ast c_{\ast\ast} \in L^\infty (\R_+\times\R_+,\mu\times \mu)$ as $\displaystyle\sup_{p>0} ||c_\ast c_{\ast\ast}||_{L^p(\R_+\times\R_+,\mu\times \mu)}=\displaystyle\sup_{p>0}||c ||^2_{L^p(\R_+,\mu)}=\displaystyle\sup_{p>0} \left(||c^p g(c)||_{L^1(\R_+)}\right)$. In fact, if a function has uniformly bounded $L^p$ norms in a finite measure space, then it is also $L^\infty$.
Moreover, they may be seen as a particular case of kernels in the form $\omega(c, c_\ast, c_{\ast\ast}) = \omega_1(c)\omega_2(c_\ast)\omega_3(c_{\ast\ast})$ with $\omega_i$ functions with certain growth (and for the symmetry assumption~\ref{ass.1} most probably $\omega_2 = \omega_3$), then:
\begin{align*}
    k(c) = \omega_1(c)\frac{\left(\int_{\R_+}\omega_2(c_\ast)k(c_\ast)\,d\mu(c_\ast)\right)\left(\int_{\R_+}\omega_3(c_{\ast\ast})\,d\mu(c_{\ast\ast})\right)}{\lambda + \frac 12\omega_1(c)\left(\int_{\R_+}\omega_2(c_\ast)\,d\mu(c_\ast)\right)\left(\int_{\R_+}\omega_3(c_{\ast\ast})\,d\mu(c_{\ast\ast})\right)}.
\end{align*}
If $\lambda < 0$, then $\lambda + \Gamma(c)/2 < \Gamma(c)$ and hence:
$$k(c) > \int_{\R_+}\frac{\omega_2(c_\ast)}{\int_{\R_+}\omega_2(c_\ast)\,d\mu(c_\ast)}k(c_\ast)\,d\mu(c_\ast)$$
and we can introduce the probability measure $\nu \coloneqq \frac{\omega_2(c_\ast)}{\int_{\R_+}\omega_2(c_\ast)\,d\mu(c_\ast)}\mu$, which is absolutely continuous with respect to $\mu$, then $k(c) > \mathbb E_\nu[k]$ for every $c \in \R_+$, where $\mathbb E_\nu$ is the expected value with respect to $\nu$. But this not possible.
\end{remark}

\section{Linear problems}
In this section we apply the methodology illustrated in the previous section to some kinetic models in opinion dynamics implementing binary interactions. As the starting point is a binary interaction, the kinetic model will not be transformed into an Allen-Cahn equation, but into a linear scattering equation. By analysing a suitable eigenvalue-eigenfunction problem we shall study the asymptotic states. %Specifically, we study the stationary state of the opinion dynamics by means of an Allen-Cahn equation.
\subsection{The Ochrombel-type simplified model}
With the same notations of the previous section, we present a further model exposed in \cite{LnMrTa}, that is the \textit{Ochrombel simplified model}, which involves binary interactions. In this model, two identical opinions are not necessary to convince an individual, but the probability that the first individual is convinced by the opinion of the second one is proportional to the connectivity of the latter and the duration of the interaction. In terms of Boltzmann-type equations, these assumptions read:
\begin{multline} \label{och_eq}
    \frac{d}{dt}\int_{\R_+}\int_{\{-1, 1\}}\phi(w, c)f(w, c, t)\,dw\,dc  \\
    = \frac 12\int_{\R_+^2}\int_{\{-1, 1\}^2}B(c, c_\ast)( \phi(w, c_\ast) - \phi(w_\ast, c_\ast))f(w, c, t)f(w_\ast, c_\ast, t)\,dw\,dw_\ast\,dc\,dc_\ast.
\end{multline}
In \cite{LnMrTa}, the choice $B(c, c_\ast) = \eta c$ is considered. Again, by deriving a closed form for the evolution equation for the total average, the authors in \cite{LnMrTa} show that the stationary average opinion may be any value in $[-1,1]$. We now want to see if the same result holds for a generic $B$. Then, again choosing $\phi(w, c) = \varphi(w)\psi(c)$ and performing the change of variables $(c, w) \mapsto (c_\ast, w_\ast)$ at the right hand side, we have the $c$-strong form for \eqref{och_eq}:
\begin{multline*} \label{och_eq_2}
    \frac{d}{dt}\int_{\{-1, 1\}}\varphi(w)f(w, c, t)\,dw \\
    = \frac 12\int_{\R_+}\int_{\{-1, 1\}^2}B(c, c_\ast)(\varphi(w_\ast) - \varphi(w))f(w, c, t)f(w_\ast, c_\ast, t)\,dw\,dw_\ast\,dc_\ast.
\end{multline*}
If $\varphi(w) = \frac{w + 1}{2}$, the evolution equation for $\tilde{p}, \tilde{q}$ defined in~\eqref{def:tildepq} is
$$\partial_t \tilde{p}(c, t) = \frac 12 \int_{\R_+}B(c, c_\ast)(\tilde{p}(c_\ast,t)\tilde{q}(c,t) - \tilde{q}(c_\ast,t)\tilde{p}(c,t))\,dc_\ast.$$
Repeating the same argument as before we obtain
\begin{align*}
    \partial_t \tilde p(c, t) &= \frac 12\int_{\R_+} B(c, c_\ast)(\tilde p(c_\ast,t)(1 - \tilde p(c,t)) - (1 - \tilde p(c_\ast,t))\tilde p(c,t)) g(c_\ast)\,dc_\ast \\
    &= \frac 12\int_{\R_+}B(c, c_\ast)(\tilde p(c_\ast,t) - \tilde p(c,t))g(c_\ast)\,dc_\ast,
\end{align*}
and, if we define 
\begin{equation}\label{def:GammaB}
\Gamma_B(c) \coloneqq \int_{\R_+}B(c, c_\ast)g(c_\ast)\,dc_\ast,    
\end{equation} 
we get%, dropping the $\tilde{ }$
\begin{equation} \label{lin_eq}
    \partial_t \tilde{p}(c, t) = \frac 12\left[\int_{\R_+} B(c, c_\ast)\tilde{p}(c_\ast,t) g(c_\ast)\,dc_\ast - \Gamma_B(c)\tilde{p}(c,t)\right],
\end{equation}
that is a linear scattering equation (or linear Boltzmann equation).
We make the following reasonable assumption.
\begin{assumption}\label{ass.3}
The kernel $B>0 $ and the distribution $g$ are defined in such a way to satisfy
\begin{equation}
B \in L^\infty(\R_+\times \R_+,\mu\times \mu). 
\end{equation}   
\end{assumption}
As a consequence, $\Gamma_B \in L^\infty (\R_+,\mu)$ and $\Gamma_B > 0$. 

Typically, we introduce
\begin{align}
      \mathcal{T}_B : L^1(\R_+,\mu) &\longmapsto   L^1(\R_+,\mu)\\
       k &\longmapsto \mathcal{T}_B k(c) := \int_{\R_+} B (c,c_\ast) k(c_\ast)  \, d\mu(c_\ast),
    \end{align}
    that is well defined, bounded and compact. The operator $\mathcal{T}_B$ allows to rewrite \eqref{lin_eq} as
\begin{equation}\label{lin_eq_2}
\partial_t \tilde{p}=\mathcal{T}_B\tilde{p}-\Gamma_B \tilde{p}.
\end{equation}
We shall denote by $\mathcal{T}_B^*$ the adjoint.
As we want to determine the stationary state of \eqref{lin_eq_2},
we introduce $n$ and $m$ the eigenfunctions of the stationary equation and of its adjoint, i.e. respectively
\begin{equation}
    \Gamma_B(c) n(c) =\mathcal{T}_B n(c),
\end{equation}
\begin{equation}
    \Gamma_B(c) m(c) =\mathcal{T}^*_B m(c).
\end{equation}
In the present case, because of the definition of $\Gamma_B$, we have that the adjoint $\mathcal{T}^*_B$ satisfies micro-reversibility, i.e.
\[
B(c,c_\ast)g(c_\ast) m(c_\ast)=B(c_\ast,c)g(c) m(c),
\]
and it is then conservative.
On the other hand, the eigenfunction of the operator $\mathcal{T}_B$ is 
\[
n=1,
\]
up to multiplicative constants (any constant is an eigenfunction). 
Specifically, we may argue that the operator 
\[
\mathcal{L}_B=\dfrac{1}{\Gamma_B}\mathcal{T}_B
\]
is linear and it is bounded. As $T_B\in L^\infty(\R_+\times \R_+,\mu\times \mu) \subset L^2(\R_+\times \R_+,\mu\times \mu)$, then the operator $\mathcal{L}_B$ is also compact.
Hence we can apply the Krein-Rutman theorem, and, therefore, there exists a positive eigenvalue (that is the spectral radius) with a positive eigenfunction that is the unique positive one (up to a multiplicative constant). We remark that $n\equiv 1$, $\delta=1$ is an eigenvalue-eigenfunction couple and, in particular, all the functions belonging to $\langle 1 \rangle\subset L^1(\R_+,\mu)$ are eigenfunctions relative to the unitary eigenvalue. The same argument can be applied to $\mathcal{T}_B^*$ and $m$ is a positive eigenfunction unique up to multiplicative constants.

Then $\tilde{p}^\infty \in \langle 1\rangle$, and thus it is a constant.
In order to detect the exact stationary state and the decay to equilibrium, we use classical arguments of General Relative Entropy~\cite{benoit}.
Firstly, in order to uniquely define the eigenfunctions, we normalize them by setting
\begin{equation}\label{eig.norm}
    \int_{\R_+} n(c) g(c) \, dc=1, \qquad \int_{\R_+} n(c) m(c) g(c) \, dc=1.
\end{equation}
Specifically, as $g$ satisfies~\eqref{def:g}, the latter implies
\[
n=1.
\]
We can state the following general result
\begin{lemma}\label{lemma1}
Let us consider Assumption~\ref{ass.3}, the normalization conditions~\eqref{eig.norm} and solutions $h(c,t)$ to~\eqref{lin_eq}, and let $H$ be a convex functional. Then we  have that
\begin{multline}\label{lemma.gen.entr}
    \dfrac{d}{dt}\int_{\R_+} m(c) n(c) H\left(\dfrac{h(c,t)}{n(c)}\right) \, g(c) \, dc\\
    =\int_{\R_+^2}B(c,c_\ast) m(c) n(c_\ast) \left[ H\left(\dfrac{h(c,t)}{n(c)}\right)-H\left(\dfrac{h(c_\ast,t)}{n(c_\ast)}\right)\right.\\
    \left.+H'\left(\dfrac{h(c_\ast,t)}{n(c_\ast)}\right)\left(\dfrac{h(c_\ast,t)}{n(c_\ast)}-\dfrac{h(c,t)}{n(c)} \right)\right] \, g(c_\ast) dc_\ast g(c) dc 
    \le 0.
\end{multline}
In particular, we have that
\begin{equation}\label{lemma.cons}
    \rho:=\int_{\R_+} m(c) h(c,t) \, g(c) dc=\int_{\R_+} m(c) h(c,0) \, g(c) dc
\end{equation}
is constant in time. Moreover,  there exists a constant $\nu_3>0$ such that
\begin{equation}\label{eq:gen.entr}
\int_{\R_+} m(c) n(c) \left(\dfrac{h(c,t)-\rho n(c)}{n(c)}\right)^2 \, g(c) dc \le \int_{\R_+} m(c) n(c) \left(\dfrac{h(c,0)-\rho n(c)}{n(c)}\right)^2 \, g(c) dc \exp^{-\nu_3 t}
  \end{equation}
\end{lemma}
\begin{proof}
The proof of this Lemma relies on classical results of General Relative Entropy as illustrated in~\cite{MICHEL2005}, that we here only adapt to the case of operators defined on the finite measure space $L^1(\R_+,\mu)$. We write it for completeness in the Appendix~\ref{appendix}. 
\end{proof}
Thanks to Lemma~\ref{lemma1}, it is possible to prove that
\begin{thm}
    Let $B$ satisfy Assumption~\ref{ass.3}, then \begin{equation}\label{rho}
        \tilde{p}^\infty=\rho, \qquad \rho=\int_{\R_+} m(c) p(c,0) \, dc 
    \end{equation}
    is an asymptotic stable equilibrium of~\eqref{lin_eq}. 
\end{thm}
The proof relies on applying Lemma \ref{lemma1} to $h(c,t)=\tilde{p}(c,t)$ whose evolution is ruled by~\eqref{lin_eq}. The constant $\rho$ is determined by~\eqref{lemma.cons} by remembering that $\tilde{p}(c,0)g(c)=p(c,0)$.
The fact that a priori any constant $\tilde{p}^\infty$ can be the stationary state, is analogous to the specific case considered in~\cite{LnMrTa}, where in the simplified model, differently with respect to the \textit{Two-against-one} model, there may be final average opinions that can be intermediate values in $[-1,1]$. The specific intermediate value is here determined by $\rho$, that is fixed by the initial condition $p(c,0)$, and by the connectivity distribution $g$ and by the kernel $B$ that determine together the operator $\mathcal{T}_B$ and, thus, $m$. Of course, in general $\tilde{p}$ is not conserved (the conserved quantity is $\hat{p}+\hat{q}$), unless $\mathcal{T}_B$ is self-adjoint. In fact, we have that $\mathcal{T}_B^*$ satisfies micro-reversibility, and therefore $\int_{\R_+} T_B(c,c_\ast) \, d\mu(c_\ast)= 1$, and it is, thus, conservative. Then, if $\mathcal{T}_B=\mathcal{T}_B^*$, we have that also $m=1$ and $\rho=\|\tilde{p}\|_{L^1_\mu }=\|p\|_{L^1}$.

In conclusion, the kinetic asymptotic stable equilibrium of~\eqref{och_eq} is
\[
f^\infty(w,c)=p^\infty(c)\delta(w-1)+q^\infty(c)\delta(w+1),
\]
where
\[
p^\infty(c)=\rho g(c) \qquad q^\infty(c)=(1-\rho) g(c),
\]
$\rho$ defined in \eqref{lemma.cons}.

\subsection{Continuous opinion model}

This method of analysis can be applied in further cases. Let us for example consider a multi-agent system where now the opinion is a continuous variable $w\in[-1,1]$ evolving as a consequence of the binary interaction rules~\cite{toscani2018PRE}
\begin{equation}\label{eq:micro}
w'= w-\nu K(c,c_\ast) (w-w_\ast)+ \xi D(w,c),    
\end{equation}
where $\nu \in (0,1/2)$ and $\xi$ is a white noise s.t. $\xi \ge \dfrac{-1+\gamma}{D}$. The latter model extends the popular interaction rules for opinion formation introduced in~\cite{toscani2006CMS}. In~\cite{toscani2018PRE}, the authors consider a multi-agent system in which the agent is characterized by the opinion $w\in[-1,1]$ and by the connectivity $c\in \R_+$ and modify the rule by adding the dependence on $c$ in $K$ and in $D$.
The evolution equation for $f(w,c,t): [-1,1]\times \R_+ \times\R_+ \rightarrow \R_+$, in weak form, is
\begin{multline} \label{och_eq_tosc}
    \frac{d}{dt}\int_{[-1, 1]}\phi(w)f(w, c, t)\,dw  \\
    = \left\langle \int_{\R_+}\int_{[-1, 1]^2}B(c, c_\ast)( \phi(w') - \phi(w))f(w, c, t)f(w_\ast, c_\ast, t)\,dw\,dw_\ast\,dc_\ast\right\rangle.
\end{multline}

In general, $f$ may be written as in~\eqref{def:fact}, where now $w\in[-1,1]$ and
we may define the average opinion given a fixed connectivity as
$$
m_w(c,t):=\int_{[-1, 1]}w h_c(w,t)  \, dw.
$$
In~\cite{toscani2018PRE}, the authors show that when there is statistical independence, i.e.~\eqref{def:indep} holds true, then $m_w$ does not depend on $c$, and the total average $M_w$ defined as
\begin{equation}
    M_w(t):= \int_{\R_+} \int_{[-1, 1]}w h_c(w,t) \, dwg(c) \, dc
\end{equation}
is conserved as $g$ is time independent. Conversely, when there is not statistical independence, then the stationary conditional density function $h_c^\infty$ can be determined by means of the quasi-invariant limit procedure. It is shown that $h_c^\infty$ has a constant (independent of $c$) average $m_w^\infty$, that is actually equal to $M_w^\infty$. Then, they remark that if the ansatz ''$m_w(c,t)$ independent of $c$'' holds true for each $t>0$, then $m_w$ and $M_w$ are conserved in time, and, then, it is possible to determine $M^\infty_w$. Else, its value can be known only if $h_c$ is known at each time or assumptions can be made on $D(w,c)$ in order to obtain an eventually nonlinear equation for $M_w$ that, however, may have a non-unique solution.

Now, setting $\phi(w)=w$ in~\eqref{och_eq_tosc}, we may find the evolution equation of $m_w(c,t)$ that is
\begin{equation} \label{cont_op}
\dfrac{\partial}{\partial t} m_w(c,t)=  \int_{\R_+}B(c, c_\ast)K(c,c_\ast)\left(m_w(c_\ast,t)-m_w(c,t)\right)g(c_\ast)\,dc_\ast,
\end{equation}
that is again a linear scattering equation of the same type of~\eqref{lin_eq} and one can repeat the same analysis, having in this case $m_w$ instead of $\tilde{p}$. 

We now define
\[
 BK(c,c_\ast):=B(c, c_\ast)K(c,c_\ast),
\]
and assume the following.
\begin{assumption}\label{ass.4}
The kernels $B>0$ and $K>0$ and the distribution $g$ are defined in such a way that
\begin{equation}\label{BK_Linf}
BK \in L^\infty(\R_+\times\R_+,\mu\times \mu).    
\end{equation}    
\end{assumption}
Then, we can define
\begin{equation}
    \Gamma_{BK}(c):= \int_{\R_+} BK (c,c_\ast) 
    \, g(c_\ast) dc_\ast,
\end{equation}
satisfying $0<\Gamma_{BK}\in L^\infty (\R_+,\mu),$
and the operator
\begin{align}
      \mathcal{T}_{BK} : L^1(\R_+,\mu) &\longmapsto   L^1(\R_+,\mu)\\
       k &\longmapsto \mathcal{T}_{BK} k(c) := \int_{\R_+} BK (c,c_\ast) k(c_\ast)  \, d\mu(c_\ast),
    \end{align}
    that is well defined, bounded and compact.
    We shall denote by $\mathcal{T}_{BK}^*$ its adjoint.
We denote by $n^{BK}$ and $m^{BK}$ the eigenfunctions of the stationary equation and of its adjoint, i.e. respectively
\begin{equation}
    \Gamma_{BK}(c) n^{BK}(c) =\mathcal{T}_{BK} n^{BK}(c),
\end{equation}
\begin{equation}
    \Gamma_{BK}(c) m^{BK}(c) =\mathcal{T}^*_{BK} m^{BK}(c).
\end{equation}
Again, because of the definition of $\Gamma_{BK}$, we have that the adjoint $\mathcal{T}^*_{BK}$ satisfies micro-reversibility, i.e.
\[
BK(c,c_\ast)g(c_\ast) m^{BK}(c_\ast)=BK(c_\ast,c)g(c) m^{BK}(c),
\]
while the eigenfunction of the operator $\mathcal{T}_{BK}$ is
\[
n^{BK}=1
\]
Again, the operator $\mathcal{T}_{BK}$ is compact and the Krein-Rutman theorem applies.
Following the same procedure as before and Lemma~\ref{lemma1}, we can state the following result.
\begin{thm}\label{thm5}
    Let $B,K$ satisfy Assumption~\ref{ass.4}. Then   
 \begin{equation}\label{rho_cont}
    m_w^\infty = \rho_w, \, \qquad \rho_w=\int_{\R_+} m^{BK}(c) m_w(c,0) g(c) \, dc.
 \end{equation} 
 is an asymptotic stable equilibrium.
\end{thm}
Again, a priori any possible constant may be the stationary state of~\eqref{cont_op}. The specific value is $\rho_w$ that is fixed by the initial condition $m_w(c,0)$,
by the connectivity distribution $g$ and by the kernel $BK$ that determines $m^{BK}$. 
With respect to the previous section, where there is a conserved quantity that is~\eqref{eq:phat+qhat}, in this case a priori there is no conserved quantity. If there is conservation, then $M_w$ is simply equal to the initial condition. In particular, here it is possible to see that $M_w$ is conserved in time  if and only if $m_w$ is a constant, i.e. it does not depend on $c$, and this is a consequence of Theorem~\ref{thm5}. Moreover, we have that 
$m_w$ is conserved if and only if $\mathcal{T}_{BK}=\mathcal{T}_{BK}^*$.
This also implies that $M_w$ is conserved, being both $M_w$ and $m_w$ constant and equal.
If there is no conservation, anyway we know the stationary asymptotic constant value $\rho_w$. %we can use the evolution equation for $m_w$, and we know that the stable equilibrium is fixed once the initial condition is known.

Then, with respect to~\cite{toscani2018PRE}, we have shown that the independence of $m_w$ with respect to $c$ is a sufficient and necessary condition for conservation. The latter is granted if $\mathcal{T}_{BK}$ is self-adjoint, as $m=n=1$, and this defines a condition on $B$ and $K$.
Else, when there is no conservation, we can determine the, anyway stable, stationary state of~\eqref{cont_op} by using~\eqref{rho_cont}, that, differently from~\cite{toscani2018PRE}, does not require the knowledge of $h_c$.

\section{Co-evolving networks}
We now consider, in the same spirit as~\cite{calzola2023}, a co-evolving network. As a consequence, we assume that the connectivity of the agent may vary in time according to a given microscopic interaction that is independent of $w$, e.g. via a process of the form
\[
dc=I(c)~dt+D(c)~dW,
\]
where $I(c)$ is a deterministic contribution, while $D(c)$ is a diffusion and $W$ is a Wiener process. For the opinion evolution we consider the  Ochrombel type simplified model and the \textit{Two-against-one} model with a kinetic version of this process. We remark that, as $c$ varies, then $g=g(c,t)$ will depend on time.

\subsection{Ochrombel simplified model}\label{sec:co_ev_och}
The weak kinetic equation for $f$ implementing the Ochrombel type simplified dynamics for the opinion and an independent microscopic dynamics for the connectivity is
\begin{multline} \label{och_eq_ev}
    \frac{d}{dt}\int_{\R_+}\int_{\{-1, 1\}}\phi(w, c)f(w, c, t)\,dw\,dc  \\
    = \frac 12\int_{\R_+^2}\int_{\{-1, 1\}^2}B(c, c_\ast)( \phi(w, c_\ast) - \phi(w_\ast, c_\ast))f(w, c, t)f(w_\ast, c_\ast, t)\,dw\,dw_\ast\,dc\,dc_\ast\\
    + \chi\left\langle\int_{\R_+}\int_{\{-1, 1\}}( \phi(w,c') -\phi(w,c))f(w, c, t)\,dw\,dc \right\rangle ,
\end{multline}
with 
\begin{equation}\label{rule.bin.c}
c'=I(c)+D(c)\xi,    
\end{equation}
being $\xi$ a Gaussian random variable (or drawn from some other noise model with mean zero), $\langle \cdot \rangle$ denotes its expectation, and $\chi$ the frequency of update of the connectivity.
Equation~\eqref{och_eq_ev} can be rewritten as
\begin{multline} \label{och_eq_ev2}
    \frac{d}{dt}\int_{\R_+}\int_{\{-1, 1\}}\phi(w, c)f(w, c, t)\,dw\,dc  \\
    = \frac 12\int_{\R_+^2}\int_{\{-1, 1\}^2}B(c, c_\ast)( \phi(w_\ast, c) - \phi(w, c))f(w, c, t)f(w_\ast, c_\ast, t)\,dw\,dw_\ast\,dc\,dc_\ast\\
    + \chi\left\langle\int_{\R_+}\int_{\{-1, 1\}}(\phi(w,c')-\phi(w,c))f(w, c, t)\,dw\,dc\right\rangle,
\end{multline}
with~\eqref{rule.bin.c}.
Concerning the evolution of the connectivity, the term related to the variation of $c$ (second term in the right-hand side in~\eqref{och_eq_ev2}) may be rewritten as~\cite{loy2020CMS}
\begin{multline} \label{och_eq_ev4}
    \frac{d}{dt}\int_{\R_+}\int_{\{-1, 1\}}\phi(w, c)f(w, c, t)\,dw\,dc  \\
    = \frac 12\int_{\R_+^2}\int_{\{-1, 1\}^2}B(c, c_\ast)( \phi(w_\ast, c) - \phi(w, c))f(w, c, t)f(w_\ast, c_\ast, t)\,dw\,dw_\ast\,dc\,dc_\ast\\
    + \chi\int_{\R_+}\int_{\{-1, 1\}}\phi(w,c)\left( \int_{\R_+}   P(c|c')f(t,w,c')dc' -f(w, c, t)\right)\,dw\,dc 
\end{multline}
where $P(c'|c)$ is a (conditional) probability density function of changing connectivity into $c'$ given the prior connectivity $c$, and satisfying
\[
\int_{\R_+} P(c'|c) dc'=1, \quad \int_{\R_+} P(c'|c) c' dc'=I(c), \quad \int_{\R_+} P(c'|c) (c'-I(c))^2 dc'=D^2(c),
\]
i.e. it has average $I(c)$ and variance $D^2(c)$.  The existence of such a probability distribution for \eqref{rule.bin.c} is an implicit assumption on $I$ and $D$.
In strong form, \eqref{och_eq_ev4} may be rewritten as
\begin{multline} \label{och_eq_ev5}
    \frac{d}{dt}\int_{\{-1, 1\}}\varphi(w)f(w, c, t)\,dw\  \\
    = \frac 12\int_{\R_+}\int_{\{-1, 1\}^2}B(c, c_\ast)( \varphi(w_\ast) - \varphi(w))f(w, c, t)f(w_\ast, c_\ast, t)\,dw\,dw_\ast\,dc_\ast\,\\
    + \chi\int_{\{-1, 1\}}\varphi(w)\left( \int_{\R_+}P(c|c')f(t,w,c')dc'-f(w, c, t)\right)\,dw .
\end{multline}
Therefore, setting $\varphi(w)=\dfrac{w+1}{2}$ and dividing by $g(c,t)$, multiplying the second term on the right-hand side by $g(c',t)/g(c',t)$ we get %and suppressing the $\tilde{ }$ above $p$, we get
\begin{multline} \label{och_eq_ev6}
    \partial_t \tilde{p}(c,t) =\frac 12\int_{\R_+} B(c, c_\ast)\tilde{p}(c_\ast,t) g(c_\ast,t)\,dc_\ast - \frac 12\,\Gamma_B(c)\tilde{p}(c,t)\\
    + \chi\left( \int_{\R_+}P(c|c')\dfrac{g(c',t)}{g(c,t)}\tilde{p}(t,c')dc'- \tilde{p}(c,t)\right)-\tilde{p}(c,t)\dfrac{\partial_t g(c,t)}{g(c,t)}.
\end{multline}
Now, setting $\varphi(w)=1$ in \eqref{och_eq_ev5}, we have that
\begin{equation}\label{eq:boltz.T}
\partial_t g(c,t)=\chi \left(\int_{\R_+}P(c|c') g(c',t) dc' -g(c,t)\right),
\end{equation}
and therefore~\eqref{och_eq_ev6} becomes
\begin{multline} \label{och_eq_ev7}
    \partial_t \tilde{p}(c,t) =\frac 12\int_{\R_+} B(c, c_\ast)\tilde{p}(c_\ast) g(c_\ast,t)\,dc_\ast - \frac 12\,\Gamma_B(c)\tilde{p}(c,t)\\
    + \chi\left( \int_{\R_+}P_g(c,c')\tilde{p}(t,c')g(c',t) \, dc'- \tilde{p}(c,t)\right)-\tilde{p}(c,t)\chi\left(\int_{\R_+}P_g(c,c') g(c',t) \, dc'-1\right),
\end{multline}
where
\begin{equation}\label{def:Pg}
P_g(c,c')= P(c|c')\dfrac{1}{g(c,t)}.
\end{equation}
Equation~\eqref{och_eq_ev6} may be rewritten in compact form as
\begin{multline} \label{och_eq_ev8}
    \partial_t \tilde{p}(c,t) =\int_{\R_+}\left[\frac 12 B(c, c_\ast)+\chi P_g(c,c_\ast)\right]\tilde{p}(c_\ast,t) g(c_\ast,t)\,dc_\ast - \left[\frac 12\,\Gamma_B(c) +\chi\Gamma_P(c)\right]\tilde{p}(c,t),
\end{multline}
where
\begin{equation}\label{def:GammaP}
\Gamma_P(c)=\int_{\R_+}P_g(c,c_\ast)g(c_\ast,t) \, dc_\ast.
\end{equation}
Again, the equation for $\tilde{p}$ is a linear scattering equation. However now, as we are dealing with a time-varying $g$, but we want to determine the stationary states of~\eqref{och_eq_ev8}, firstly we need to identify the asymptotic limit $g^\infty$. In the present notation, we report here the result given in~\cite{bisi2024Physd}.
\begin{thm}
Let the mapping $c_\ast\mapsto P(\cdot\,\vert\,c_\ast)$ be Lipschitz continuous with respect to the $1$-Wasserstein metric $W_1$ in the space of the probability measures on $\R_+$ , i.e. let a constant $\textrm{Lip}{(P)}>0$ exist such that
\begin{equation}\label{propr.PLip}
W_1\bigl(P(\cdot\,\vert\,c_\ast),\,P(\cdot\,\vert\,d_\ast)\bigr)\leq
	\textrm{Lip}{(P)}\abs{d_\ast-c_\ast},
		\qquad \forall\,\,c_\ast,\,d_\ast\in\R_+. 
  \end{equation}
If $\chi\textrm{Lip}{(P)}<\frac{1}{2}$ then~\eqref{eq:boltz.T} admits a unique equilibrium distribution $g_\infty$, which is a probability measure on $\R_+$ and which is also globally attractive, i.e.
$$ \lim_{t\to +\infty}W_1(g(\cdot, \, t),\,g^\infty)=0 $$
for every solution $g$ to~\eqref{eq:boltz.T}.
\end{thm}
Then, we introduce the finite measure space related to the stationary probability density function $g^\infty$
\[
L^1(\R_+,\mu^\infty), \qquad d\mu^\infty =g^\infty(c) dc.
\]

We can remark that the stationary solutions $\tilde{p}^\infty$ of~\eqref{och_eq_ev8} are given by
\[
\left[\frac 12\,\Gamma_B^\infty(c) +\chi\Gamma_P^\infty(c)\right]\tilde{p}^\infty(c)=\mathcal{T}_{BP} \tilde{p}^\infty(c) 
\]
where
\begin{align}
      \mathcal{T}_{BP} : L^1(\R_+,\mu^\infty) &\longmapsto   L^1(\R_+,\mu^\infty)\\
       k &\longmapsto \mathcal{T}_{BP} k(c) := \mathcal{T}_B^\infty k(c) + \mathcal{T}_P k(c)
    \end{align}
    and
\begin{align}\label{def_TP}
      \mathcal{T}_{P} : L^1(\R_+,\mu^\infty) &\longmapsto   L^1(\R_+,\mu^\infty)\\
       k &\longmapsto \mathcal{T}_{P} k(c) := \int_{\R_+}P_g(c,c_\ast) k(c_\ast) \, g^\infty(c_\ast) dc_\ast
    \end{align}    
%\[
%T_{BP}(c,c_\ast)= \frac 12 B(c, c_\ast)+\chi P^\infty_g(c, c_\ast),
%\]
and being $P^\infty_g, \Gamma_B^\infty, \Gamma_P^\infty$ defined by~\eqref{def:Pg},\eqref{def:GammaB},\eqref{def:GammaP} with $g=g^\infty$.
%Therefore, we have that
%\[
%\int_{\R_+} T_{BP}(c,c_\ast) \, d\mu^\infty(c_\ast)= 1. 
%\]
We consider again Assumption~\ref{ass.3}, and, similarly, also the following.
\begin{assumption}\label{ass.5}
We assume that $P$ and $g$ are defined in such a way that
\begin{equation}\label{Pg_infty}
P_g^\infty \in L^\infty(\R_+\times\R_+, \mu \times \mu), \quad P_g^\infty >0 \, .   
\end{equation}    
\end{assumption}
We again introduce the eigenfunctions $n^{BP}, m^{BP}$ of the operators $\mathcal{T}_{BP}$ and $\mathcal{T}^*_{BP}$, and we have again that $n^{BP}=1$ and $m^{BP}$ is the unique positive eigenfunction determined by the normalization condition $\int_{\R_+} n^{BP}(c) g^\infty(c) \, dc=\int_{\R_+} m^{BP}(c) g^\infty(c) \, dc=1$.

Following the same steps as in Section 3, again we can state the following result.
\begin{thm}
    Let $P$ satisfy~\eqref{propr.PLip} with $\chi\textrm{Lip}(P) <\frac{1}{2}$, and let $B,P_g^\infty$ satisfy Assumptions~\ref{ass.3},\ref{ass.5}. Then the asymptotic stationary equilibrium of~\eqref{och_eq_ev8} is
    \[
    \tilde{p}^\infty=\rho_g, \qquad \rho_g=\int_{\R_+} m^{BP}(c) p(c,0) \, dc.
    \]
\end{thm}

Remembering that the constant steady state is actually $\tilde{p}^\infty$ defined in~\eqref{def:tildepq}, going back to the original definition of $p$~\eqref{eq:f}, we then have that the stationary distribution of the connectivity of the agents having opinion $w=1$ is
\begin{equation}
    p^\infty(c)=g^\infty(c) \rho_g.
\end{equation}
Again, $\tilde{p}^\infty$ is conserved if $\mathcal{T}_{BP}$ is self-adjoint. However, we know the stationary solution $\rho_g$ if $m^{BP}$ is determined and the initial condition $p(c,0)$ is given. %Else, we need to solve the evolution equation~\eqref{och_eq_ev8}. Here, the problem is that the time evolution of $g(c,t)$ should be known. As such, for the moment we leave it as an open problem.

\subsection{Two-against-one model}
We now consider the Two-against-one model on a co-evolving network.
Exploiting the previous computations, the $c$-strong kinetic equation for $f$ implementing the Two-against-one dynamics for the opinion and the same dynamics for the connectivity is
\begin{multline}
	\frac{d}{dt}\int_{\{-1,\,1\}}\varphi(w)f(w,c,t)\,dw 
	=\frac{\eta}{3}\int_{\R_+^2}\int_{\{-1,\,1\}^3}\omega(c,c_\ast,c_{\ast\ast})\bigl(\varphi(w_{\ast\ast})-\varphi(w)\bigr) \\
		\times f(w,c,t)f(w_{\ast\ast},c_\ast,t)f(w_{\ast\ast},c_{\ast\ast},t)\,dw\,dw_\ast\,dw_{\ast\ast}\,dc_\ast\,dc_{\ast\ast} \\
    + \chi\int_{\{-1, 1\}}\varphi(w)\left( \int_{\R_+}P(c|c')f(t,w,c')dc'-f(w, c, t)\right)\,dw.
	\label{eq:Boltz_c_strong_coev}
\end{multline}
Then, following the same procedure as in Sec.~\ref{sec:allen_cahn} and~\ref{sec:co_ev_och}, we find the evolution equation for the quantity~\eqref{def:tildepq}, that, dropping the$~\tilde{ }$ is
\begin{equation}\label{coev_allen_cahn}
   \partial_t p(c,t)=\mathcal{D}(p)+\Gamma(c)\left[p(1-p)(2p-1)\right]+\int_{\R_+}\chi P_g(c,c_\ast)p(c_\ast,t) g(c_\ast,t)\,dc_\ast - \chi\Gamma_P(c)p(c,t) 
\end{equation}
where $\mathcal{D}$ is defined as in~\eqref{def.D}.
The stationary states $p_i$ of~\eqref{coev_allen_cahn} are the solutions of
\begin{equation}\label{coev_stat}
  \mathcal{D}(p_i)+\Gamma(c)\left[p_i(1-p_i)(2p_i-1)\right]+\int_{\R_+}\chi P_g^\infty(c,c_\ast)p_i(c_\ast,t) g(c_\ast,t)\,dc_\ast - \chi\Gamma^\infty_P(c)p_i(c,t)=0.  
\end{equation}
We remark that, under the assumption~\eqref{Pg_infty}, following Sec.~\ref{sec:allen_cahn}, we find that $p_0=0, p_1=1, p_2=\frac{1}{2}$ are constant solutions of~\eqref{coev_stat}. Repeating the same linear stability analysis as in Sec.~\ref{sec:lin_stab}, we find that $p_0=0, p_1=1$ are stable stationary states. On the other hand, for the stationary state $p_2$, we have that the evolution equation of the perturbation $k_2$ is
\begin{multline*}
\partial_t k_2=\int_{\R_+^2}\dfrac{\omega(c,c_\ast,c_{\ast\ast})}{2}g(c_\ast)g(c_{\ast\ast})\big(k_2(c_\ast,t)+k_2(c_{\ast\ast},t)-2k_2(c,t)\big) \, dc_\ast \, dc_{\ast\ast}+\dfrac{\Gamma(c)}{2}k_2\\
+\int_{\R_+}\chi P_g(c,c_\ast)k_2(c_\ast,t) g(c_\ast,t)\,dc_\ast - \chi\Gamma_P(c)k_2(c,t).
\end{multline*}
If we now consider $k_2= e^{\lambda t} k(c)$, the latter becomes
\begin{multline}\label{def.eigenvalue_pb_coev}
    \lambda k(c)=\int_{\R_+^2}\omega(c,c_\ast,c_{\ast\ast})g(c_\ast)g(c_{\ast\ast})k(c_\ast)\, dc_\ast \, dc_{\ast\ast}-\dfrac{\Gamma(c)}{2}k(c)\\
    +\int_{\R_+}\chi P_g(c,c_\ast)k(c_\ast) g(c_\ast,t)\,dc_\ast - \chi\Gamma_P(c)k(c),
\end{multline}
where we have used Assumption~\ref{ass.1}. We now introduce
\begin{equation}
\begin{aligned}
    &\mathcal{B}_\lambda:L_1(\R_+,\mu^\infty) \longrightarrow L_1(\R_+,\mu^\infty)\\
    &\mathcal{B}_\lambda k(c):= \dfrac{\displaystyle \int_{\R_+} G^\infty(c,c_\ast)  g^\infty(c_\ast) k(c_\ast) \, dc_{\ast} +\chi \int_{\R_+} P_g^\infty (c,c_\ast) g^\infty(c_\ast) k(c_\ast) \, dc_{\ast}}{\lambda+\dfrac{\Gamma^\infty(c)}{2}+\chi \Gamma_P^\infty(c)},
\end{aligned}
\end{equation}
where the apex $^\infty$ indicates the corresponding quantity defined with $g=g^\infty$. 
%We define
%\begin{align}
%      \mathcal{T}_{P} : L^1(\R_+,\mu^\infty) &\longmapsto   L^1(\R_+,\mu^\infty)\\
%       k &\longmapsto \mathcal{T}_{P} k(c) := \int_{\R_+} P_g^\infty (c,c_\ast) k(c_\ast)  \, d\mu^\infty(c_\ast).
%    \end{align}
%    Again, we assume 
%\begin{equation}
%    G^\infty \in L^\infty(\R_+\times \R_+,d\mu^\infty\times d\mu^\infty), \quad G^\infty>0 \, a.e.
%\end{equation}
%\begin{equation}
%    P_g^\infty \in L^\infty(\R_+\times \R_+,d\mu^\infty\times d\mu^\infty), \quad P_g^\infty>0 \, a.e.
%\end{equation}
As a consequence of Assumptions~\ref{ass.2},\ref{ass.5}, we have that
\[
 \dfrac{G^\infty(c,c_\ast)  +\chi  P_g^\infty (c,c_\ast)}{\lambda+\dfrac{\Gamma^\infty(c)}{2}+\chi \Gamma_P^\infty(c)} \in L^\infty(\R_+\times \R_+,\mu^\infty\times \mu^\infty).
\]
Following the same arguments as in Section~\ref{sec:lin_stab}, we can apply the Krein-Rutman Theorem and we remark that for $k \equiv 1$, $\delta(\lambda)=1$ is an eigenvalue-eigenfunction couple if 
\[
\int_{\R_+ }\dfrac{\Gamma^\infty(c)+\chi \Gamma_P^\infty(c)}{\lambda +\Gamma^\infty(c)/2+\chi \Gamma_P^\infty(c)} \, g^\infty(c) dc =1.
\]
Arguing like in Section~\ref{sec:lin_stab}, we can see that the latter is not verified if $\lambda <0$.

We can summarize these observations in one theorem.
\begin{thm} \label{thm_pos}
    Let $\omega$ satisfy Assumptions~\ref{ass.1}-\ref{ass.2}, $B$ satisfy assumption~\ref{ass.3}, $P$ satisfy~\eqref{propr.PLip} and $P_g^\infty$ satisfy Assumption~\ref{ass.5}.
    %Let the kernel $G^\infty$ be defined by \eqref{def_G}, $P_g^\infty$ defined by \eqref{def:Pg}, $\Gamma^\infty, \Gamma_P^\infty$ be defined by~\eqref{def:GammaB}-\eqref{def:GammaP} respectively, with $g=g^\infty$. %satisfy \eqref{propr_Gl}
    Then the operator $\mathcal D'\left(\frac 12\right) + \frac \Gamma 2 Id +\chi \mathcal{T}_P-\chi Id$, where $\mathcal{T}_P$ is defined in~\eqref{def_TP}, has a positive eigenvalue and this means that the stationary state $\tilde{p}_2 = \frac 12$ of equation \eqref{coev_allen_cahn} is linearly unstable.
\end{thm}

%%%%%%%%%%%%%%%%%%%%%%%%%%%%%%%%%%%%%%%%%%%%%%%%%%%%%%%%%%%%%%%%%%%%%%%%%%%%%%%%
%\newpage

\section{Numerical tests}
We now present some numerical tests, in order to illustrate the asymptotic behavior of the discussed models, both in the case of \textit{Tow-against-one} and \textit{Ochrombel simplified} microscopic interactions. Moreover, we shall consider the case of the continuous opinion dynamics.

\subsection{Numerical scheme}
In order to integrate numerically Eq. \eqref{eq:Allen_Cahn} we consider its equivalent formulation
\[
\partial_t p(c)=\int_{\R_+^2}\omega(c,c_\ast,c_{\ast\ast})g(c_\ast)g(c_{\ast\ast}) \left[ (1-p(c))p(c_\ast)p(c_{\ast\ast})-p(c)(1-p(c_\ast))(1-p(c_{\ast\ast}))\right]\, dc_\ast \, dc_{\ast\ast}
\]
where we have dropped the dependence of $p$ on $t$.
Let us consider a time discretization $\lbrace t^n \rbrace_n$ where $dt:=t^{n+1}-t^n$ is the time step. Hence, denoting in a semi-continuous notation $p^n(c)$ an approximation of $p(c,t^n)$, we can consider the following explicit-implicit scheme
\[
\dfrac{p^{n+1}(c)-p^n(c)}{dt}=(1-p^{n+1}(c))A^n - p^{n+1}(c)B^n
\]
where 
\[
A^n:=\int_{\R_+^2}\omega(c,c_\ast,c_{\ast\ast})g(c_\ast)g(c_{\ast\ast})p^n(c_\ast)p^n(c_{\ast\ast})\, dc_\ast \, dc_{\ast\ast}
\]
and
\[
B^n:=\int_{\R_+^2}\omega(c,c_\ast,c_{\ast\ast})g(c_\ast)g(c_{\ast\ast}) (1-p^n(c_\ast))(1-p^n(c_{\ast\ast}))\, dc_\ast \, dc_{\ast\ast}.
\]
These means that $$p^{n + 1}(c) = \frac{p^n(c) + dt\,A^n}{1 + dt(A^n + B^n)}.$$ One can proceed likewise with equation \eqref{och_eq}.

\subsection{Two-against one model}
In this subsection, we illustrate the asymptotic trend of the Allen-Cahn equation~\eqref{eq:Allen_Cahn} for $\tilde{p}$. In Fig.~\ref{fig.1}, we consider two different pairs of connectivity distributions $g$ and interaction kernels $\omega$. In both cases we consider a constant initial $\tilde{p}(c,0) = $ and, specifically, we selected different $\tilde{p}(c,0)$ in a neighbourhood of $0.5$, that is the coexistence unstable asymptotic equilibrium. Then, in $(a)$: $g(c) = \frac c2\mathbbm 1\{0 \leq c \leq 2\}$ and $\omega(c, c_\ast, c_{\ast\ast}) = cc_\ast c_{\ast\ast}$, while in $(b)$ $g(c) = 0.5\,\mathbbm 1\{0 \leq c \leq 2\}$ and $\omega(c, c_\ast, c_{\ast\ast}) = \mathbb 1\{|c_\ast - c_{\ast\ast}| \leq 1\}$. In both cases we remark that Assumption~\ref{ass.2} holds as $g$ has compact support. In both cases, we remark that the unstable constant $\tilde{p}_2=\frac{1}{2}$ is the stationary state only if it is the initial condition, i.e. $\tilde{p}(c,0)=\frac{1}{2}$. Else, we remark that the stationary state is either $+1$ or $-1$ according to the value of the initial constant condition. In Fig.~\ref{fig:2}, the same choice of the connectivity distribution $g(c) = 0.1e^{-0.1c}$ and interaction kernel $\omega(c, c_\ast, c_{\ast\ast}) = c\mathbb 1\{|c_\ast - c_{\ast\ast}|$ is made. Then, in $(a)$ the initial datum is a non rescaled truncated Gaussian distribution centered in $1.5$ with unitary variance, whereas in $(b)$ the same initial distribution is normalised. Again, we remark that convergence is towards one of the two polarized configurations.
\begin{figure} 
    \centering
    \begin{subfigure}[]%{0.49\textwidth}
        \centering
        \includegraphics[width=0.49\textwidth]{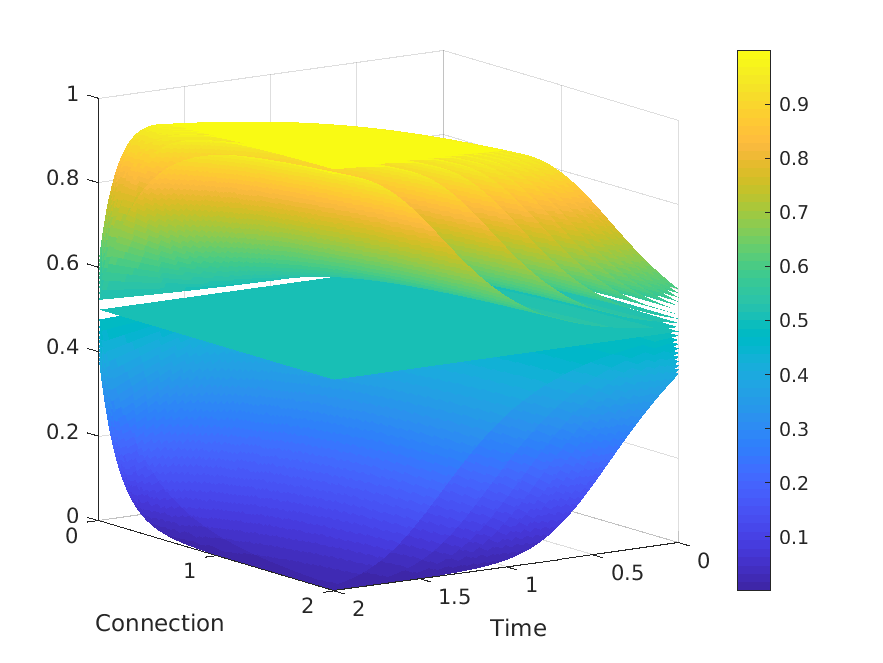}
    \end{subfigure}
    \begin{subfigure}[]%{0.49\textwidth}
        \centering
        \includegraphics[width=0.49\textwidth]{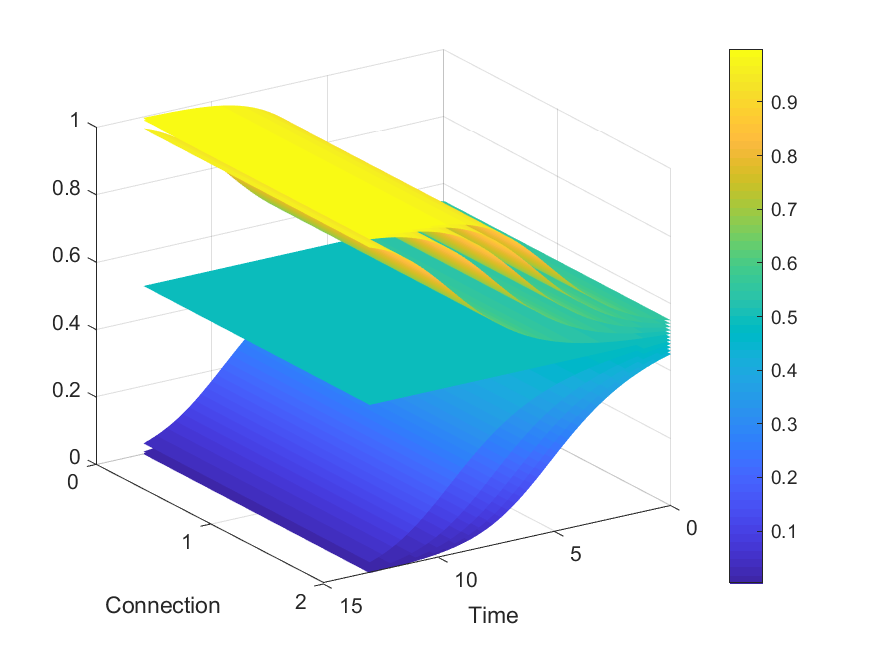}
    \end{subfigure}
    \caption{Time evolution of the numerical simulations of the solution $\tilde{p}(c,t)$ to the model given by \eqref{eq:Allen_Cahn} with $\tilde{p}(c,0) = $constant. In both cases, we selected different $\tilde{p}(c,0)$ in a neighbourhood of $0.5$. In $(a)$: $g(c) = \frac c2\mathbbm 1\{0 \leq c \leq 2\}$ and $\omega(c, c_\ast, c_{\ast\ast}) = cc_\ast c_{\ast\ast}$, while in $(b)$ $g(c) = 0.5\,\mathbbm 1\{0 \leq c \leq 2\}$ and $\omega(c, c_\ast, c_{\ast\ast}) = \mathbb 1\{|c_\ast - c_{\ast\ast}| \leq 1\}$. }
    \label{fig.1}
\end{figure}
\begin{figure} 
    \centering
    \begin{subfigure}[]%{0.49\textwidth}
        \centering
        \includegraphics[width=0.49\textwidth]{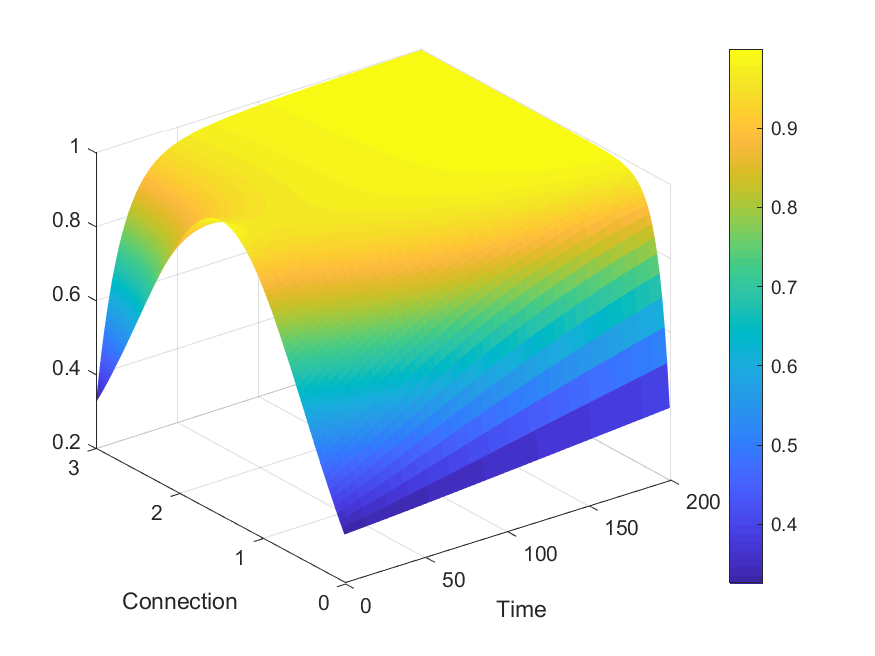}
    \end{subfigure}
    \begin{subfigure}[]%{0.49\textwidth}
        \centering
        \includegraphics[width=0.49\textwidth]{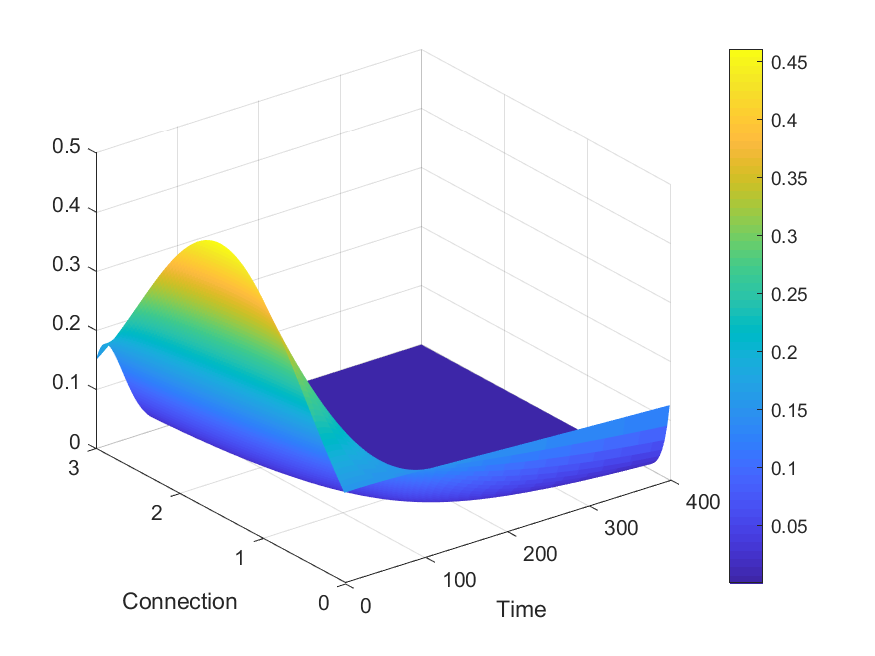}
    \end{subfigure}
    \caption{Time evolution of the numerical simulations of the solution $\tilde{p}(c,t)$ to the model given by \eqref{eq:Allen_Cahn} with, $g(c) = 0.1e^{-0.1c}$ and $\omega(c, c_\ast, c_{\ast\ast}) = c\mathbb 1\{|c_\ast - c_{\ast\ast}| \leq 1\}$. In the first case the initial datum is a non rescaled truncated Gaussian distribution centered in 1.5 with unitary variance, whereas in the second figure the same distribution is normalised.}
    \label{fig:2}
\end{figure}

\subsection{Ochrombel simplified model}
In this subsection we consider the evolution equation~\eqref{lin_eq}, that corresponds to the Ochrombel simplified model~\eqref{och_eq}. The results are reported in Figure~\ref{fig:3}. As distribution of the connectivity, we consider $g(c) = e^{-c}$, while the interaction kernel is proportional to the interactions of both agents: $B(c, c_\ast) = cc_\ast$. We remark that for the properties of $g$, Assumption~\ref{ass.3} holds. Two different initial data are considered: in $(a)$ $\tilde{p}(c,0) = e^{- c}$, while in $(b)$ $\tilde{p}(c,0)$ is a truncated (on the positive axis) Gaussian centered in $2$ with unitary variance in the second one, i.e. $\tilde{p}(c,0) =  e^{-(c - 2)^2}$. In both cases the asymptotic value~\eqref{rho} is recovered: while in $(a)$ $\rho = \int_{\R_+} e^{-2c}\,dc = \frac 12$, in $(b)$ we have that $\rho = \frac 1{\int_{\R_+}e^{-\frac 12(c - 2)^2}\,dc}\int_{\R_+}e^{-\frac 12(c - 2)^2 - c}\,dc \simeq 0.19$.
\begin{figure} 
    \centering
    \begin{subfigure}[]%{0.49\textwidth}
       \centering
       \includegraphics[width=0.49\textwidth]{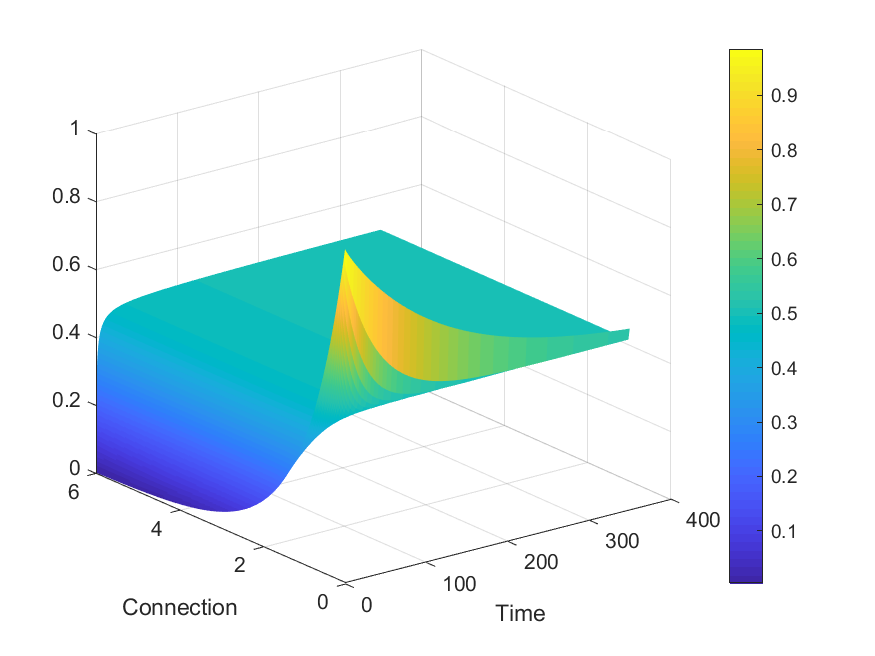}
    \end{subfigure}
    \begin{subfigure}[]%{0.49\textwidth}
        \centering
        \includegraphics[width=0.49\textwidth]{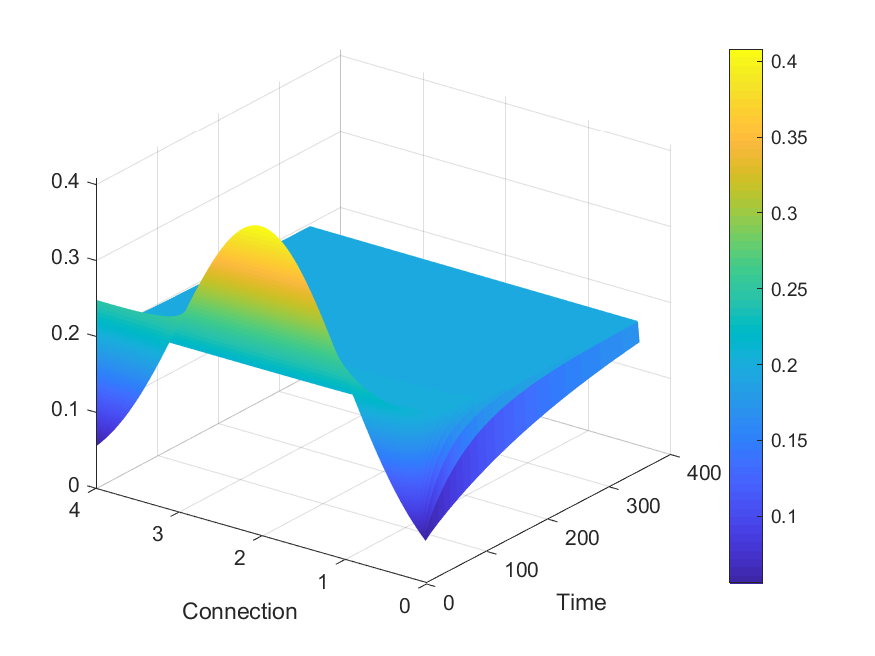}
    \end{subfigure}
    \caption{Time evolution of the numerical simulations of the solution $\tilde{p}(c,t)$ to the model given by~\eqref{lin_eq} with $g(c) = e^{-c}$, $B(c, c_\ast) = cc_\ast$, $\tilde{p}(c,0) = e^{- c}$ in the first figure and $\tilde{p}(c,0)$ truncated (on the positive axis) Gaussian centered in $2$ with unitary variance in the second one. One can recover the respective asymptotic values $\rho = \int_{\R_+} e^{-2c}\,dc = \frac 12$ and $\rho = \frac 1{\int_{\R_+}e^{-\frac 12(x - 2)^2}\,dc}\int_{\R_+}e^{-\frac 12(c - 2)^2 - c}\,dc \simeq 0.19$.}
    \label{fig:3}
\end{figure}

\subsection{Continuous opinion model}
In this subsection, we present the numerical solution of~\eqref{cont_op} and its stationary solution. The results are shown in Figure~\ref{fig:4}. Here $B = 1$, $K(c, c_\ast) = \frac{c_\ast}{c + c_\ast}$ initial datum $m_w(c, 0) = e^{-c}$ as in~\cite{toscani2018PRE}. The choice of $K$ and $B$ corresponds to a left eigenfunction $m^{BK}=c$. In $(a)$ the distribution of the connectivity is  $g(c) = e^{-c}$, that has thin tails, while in $(b)$ a fat-tailed distribution is chosen: $g(c) = \frac{e^{-\frac{1}{c}}}{c^2}$ (the same as in~\cite{toscani2018PRE}). We show that we recover the respective asymptotic values $\rho_w$ defined by~\eqref{rho_cont}, that are $\rho_w=\int_{\R_+} c e^{-2c} \, dc=0.25$ in $(a)$, and $\rho_w=\int_{\R_+} \dfrac{e^{-\frac{1 + c^2}{c}}}{c^2} \, dc \simeq 0.28$ in $(b)$.
\begin{figure} 
    \centering
    \begin{subfigure}[]%{0.49\textwidth}
       \centering
       \includegraphics[width=0.49\textwidth]{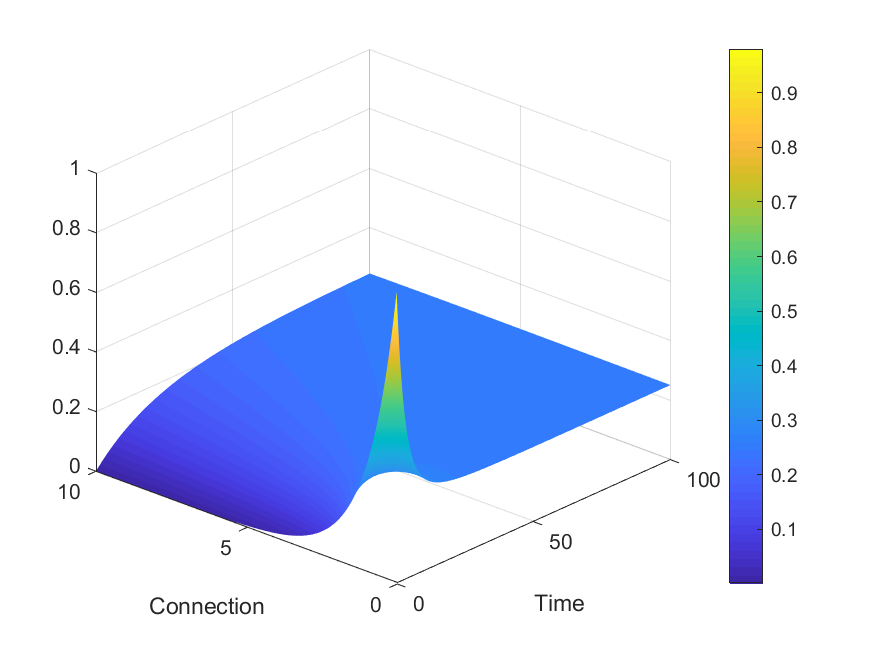}
    \end{subfigure}
    \begin{subfigure}[]%{0.49\textwidth}
        \centering
        \includegraphics[width=0.49\textwidth]{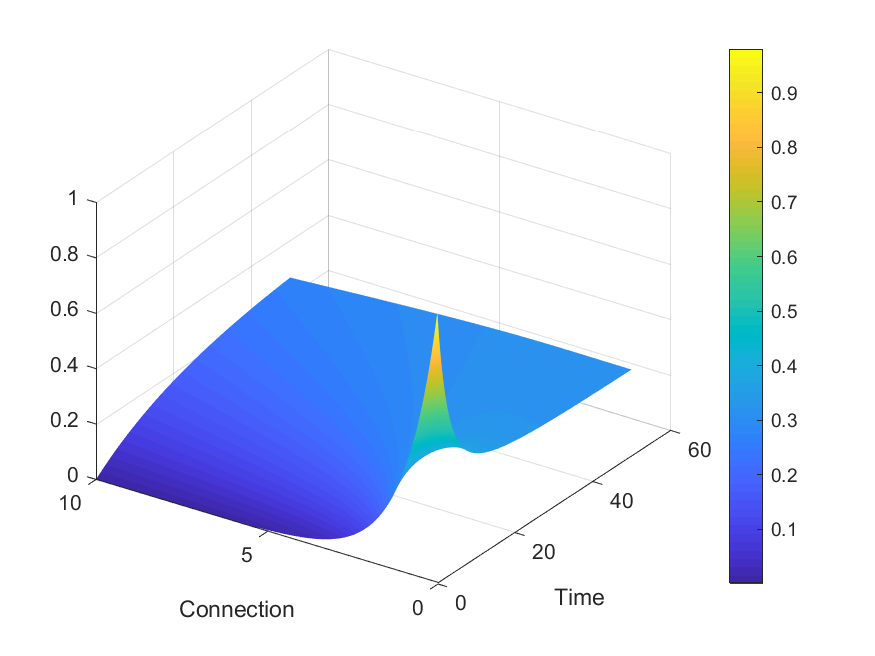}
    \end{subfigure}
    \caption{Time evolution of the numerical simulations of the solution $m_w(c,t)$ to equation \eqref{cont_op} with $B = 1$, $K(c, c_\ast) = \frac{c_\ast}{c + c_\ast}$, initial datum $m_w(c, 0) = e^{-c}$. In $(a)$  $g(c) = e^{-c}$, while in $(b)$ $g(c) = \frac{e^{-\frac{1}{c}}}{c^2}$. One can recover the respective asymptotic values $\rho_w=\int_{\R_+} c e^{-2c} \, dc=0.25$ and $\rho_w=\int_{\R_+} \dfrac{e^{-\frac{1 + c^2}{c}}}{c^2} \, dc \simeq 0.28$.}
    \label{fig:4}
\end{figure}

\section{Conclusion}
In this paper, we have studied the asymptotic stationary equilibria of Boltzmann-type equations for opinion formation on a social network. The microscopic dynamics prescribes for each agent a degree of connectivity and a binary value of the opinion evolving according to majority-like rules. Specifically, we have considered a Boltzmann-type kinetic model for the Two-against-one model that involves ternary interactions. In this case we have derived an Allen-Cahn equation for the time-evolution of the fraction of individuals with a given connectivity holding one of the two opposite opinions. The linear stability analysis of the Allen-Cahn equation is performed by solving an eigenfunction-eigenvalue problem related to an integral operator defined by the interaction kernel on a suitable finite measure space defined by the distribution of the connectivity. The linear stability analysis of the Allen-Cahn equation has allowed us to conclude that the polarised configurations are always stable, while the coexistence configuration is always unstable. This is true for any choice of the interaction kernel that gives the rate at which individuals interact according to their degree of connectivity.

We have then considered the Ochrombel-simplified model, that relies on binary interactions. In this case, the evolution equation for the fraction of individuals with a given connectivity holding one of the two opinions, is a linear scattering equation. Also in this case, the study of the stationary equilibria amounts to the analysis of an eigenvalue-eigenfunction problem for a suitable integral operator related to the interaction kernel and set on a finite measure space defined by the distribution of the connectivity. In this case, the possible stationary solutions correspond to an average opinion that may be any value in $[-1,1]$, even though the possible opinion states are $\pm 1$. In this case, then, coexistence is possible for any interaction kernel. Moreover, thanks to General Relative Entropy results, the exact asymptotic stationary stable equilibrium could be determined.

We have then considered considered a model for opinion formation on a social network when the microscopic opinion value is a continuous one that ranges in $[-1,1]$ as done in~\cite{toscani2018PRE}. The application of the same methodology as in the Ochrombel case has allowed us to determine the value of the asymptotic stationary stable average opinion for a fixed degree of connectivity, by only considering regularity assumptions on the interaction kernels and not on the dependence of the average opinion on the degree of connectivity. Specifically, this methodology allows to determine the asymptotic average opinion also in the case where it is not constant with respect to $c$ and without the need of determining it from the conditioned kinetic distribution of the opinion $h_c(w)$.

Eventually, we have applied the methodology to the case of a co-evolving network, by considering both the case of the Two-against-one model and the Ochrombel simplified model.
In all cases, the regularity assumptions on the interaction kernels on the finite measure space defined by the distribution of the connectivity are minimal ($L^\infty$) and depend on the interplay between the interaction kernels ($\omega, B, K$) and $g$.

The interest of the present manuscript relies on the possibility of determining rigorously the stationary equilibria and their stability of an opinion formation problem that is formerly described by a kinetic Boltzmann--type model, which can be derived from appropriate microscopic stochastic processes that can model accurately the interactions of the agents on a social network, which, in the present case, is described statistically.

Possible future works include the study of the continuous opinion model on co-evolving networks and the study of the equilibria of an opinion formation process (both in the case of a discrete and a continuous opinion) on a finite graph, a model in the same spirit as the one proposed in~\cite{loy2021MBE,loy2021KRM} in the case of the diffusion of a virus.

\section*{Acknowledgements}
MB acknowledges support from DESY (Hamburg, Germany),
a member of the Helmholtz Association HGF.
NL is member of INdAM-GNFM. NL acknowledges support from GNFM (Gruppo Nazionale per la Fisica Matematica) of IN-
dAM (Istituto Nazionale di Alta Matematica) through a “Progetto Giovani 2020” grant. This work was also
partially supported by the Italian Ministry of University and Research (MUR) through the “Dipartimenti di Ec-
cellenza” Programme (2018-2022), Department of Mathematical Sciences “G. L. Lagrange”, Politecnico di Torino
(CUP: E11G18000350001). This work is also part of the activities of the PRIN 2020 project (No. 2020JLWP23)
“Integrated Mathematical Approaches to Socio–Epidemiological Dynamics”.

\section{Appendix}\label{appendix}
\begin{proof}
It is straightforward to verify that if $h(c,t)$ is a solution of~\eqref{lin_eq}, then
\begin{multline}\label{lemma.entropy}
   \partial_t \left( m(c) n(c) H\left(\dfrac{h(c,t)}{n(c)} \right)\right)\\
+\int_{\R_+}B(c,c_\ast)\left[m(c_\ast) n(c) H\left(\dfrac{h(c,t)}{n(c)}\right)-m(c)n(c_\ast)H\left(\dfrac{h(c_\ast,t)}{n(c_\ast)}\right)\right] \, g(c_\ast) dc_\ast \\
   = \int_{\R_+}B(c,c_\ast) m(c) n(c_\ast) \left[ H\left(\dfrac{h(c,t)}{n(c)}\right)-H\left(\dfrac{h(c_\ast,t)}{n(c_\ast)}\right)+H'\left(\dfrac{h(c_\ast,t)}{n(c_\ast)}\right)\left(\dfrac{h(c_\ast,t)}{n(c_\ast)}-\dfrac{h(c,t)}{n(c)} \right)\right] \, g(c_\ast) dc_\ast
\end{multline}
and, after integration on $g(c)dc$, ~\eqref{lemma.gen.entr} holds.

Choosing $H(u)=u$ (that is convex) allows to see that
\[
\dfrac{d}{dt}\int_{\R_+} m(c) h(c,t)  \, g(c) \, dc=0,
\]
and therefore~\eqref{lemma.cons} holds.

Then, we choose $H(u)=u^2$ and we work on $h(c,t)=p(c,t)-\rho n(c)$. We remark that $\int_{\R_+} m(c) h(c) \, g(c) dc=0$. 
We have that
\[
\dfrac{d}{dt}\int_{\R_+} m(c) n(c) H\left(\dfrac{p(c,t)}{n(c)}\right) \, g(c) \, dc=-\int_{\R_+^2} B(c_\ast,c) m(c) n(c) \left[\dfrac{h(c)}{n(c)}-\dfrac{h(c_\ast)}{n(c_\ast)} \right]^2 \, g(c_\ast) dc_\ast g(c) dc
\]
It is then possible to prove a Poincaré inequality
\[
\nu_3 \int_{\R_+} m(c) n(c) \left(\dfrac{h}{n}\right)^2 g(c) dc \le \int_{\R_+^2} B(c_\ast,c) m(c) n(c) \left[\dfrac{h(c)}{n(c)}-\dfrac{h(c_\ast)}{n(c_\ast)} \right]^2 \, g(c_\ast) dc_\ast g(c) dc
\]
that holds for any $\psi>0$ as long as it is possible to define
\begin{equation}\label{nu}
\nu_1=\int_{\R_+} n(c) m^2(c)/\psi(c) \, g(c) dc <\infty \qquad \nu_2 \psi(c_\ast) n(c) \le B(c,c_\ast), \qquad \nu_3=(\nu_1 \nu_2)^{-1}. 
\end{equation}
In fact, whenever $\int_{\R_+} m(c) h(c,t) g(c) \, dc=0$, then
\begin{multline*}
\int_{\R_+} m(c) n(c) \left(\dfrac{h}{n}\right)^2 g(c) dc \le\\
\nu_1 \int_{\R_+}\int_{\R_+} m(c) n(c) \left[\dfrac{h(c,t)}{n(c)}-\dfrac{h(c_\ast,t)}{n(c_\ast)}\right]^2 \psi(c_\ast) n(c_\ast) g(c_\ast) \, dc_\ast g(c) \, dc\\
\le \nu_1 \nu_2\int_{\R_+}\int_{\R_+} m(c) n(c_\ast) B(c,c_\ast) \left[\dfrac{h(c,t)}{n(c)}-\dfrac{h(c_\ast,t)}{n(c_\ast)}\right]^2 g(c) g(c_\ast) \,  dc_\ast dc
\end{multline*}
where the first inequality follows from the Holder inequality on the finite measure space $L^2(\R_+,\mu), L^1(\R_+,\mu)$ and from the first in~\eqref{nu}, while the second inequality follows from the second in~\eqref{nu}.
Specifically, we can choose
$$
\psi(c_\ast)=\displaystyle\min_c B(c,c_\ast), \quad \nu_2 \in (0,1),
$$
remembering that $B>0$.
Then, 
\[
\dfrac{d}{dt}\int_{\R_+} m(c) n(c) H\left(\dfrac{p(c,t)}{n(c)}\right) \, g(c) \, dc \le -\nu_3 \int_{\R_+} m(c) n(c) \left(\dfrac{h}{n}\right)^2 g(c) dc
\]
and~\eqref{eq:gen.entr} follows by a simple use of Gronwall lemma.
\end{proof}

\bibliographystyle{plain}
\bibliography{LnRmTa-phase_transition}

\begin{thebibliography}{10}

\bibitem{albert2002RMP}
R.~Albert and A.-L. Barab\'{a}si.
\newblock Statistical mechanics of complex networks.
\newblock {\em Rev. Modern Phys.}, 74(1):1--47, 2002.

\bibitem{calzola2023}
G.~Albi, E.~Calzola, and G.~Dimarco.
\newblock A data-driven kinetic model for opinion dynamics with social network
  contacts.
\newblock {\em European Journal of Applied Mathematics}, pages 1--27, 2024.

\bibitem{albi2017KRM}
G.~Albi, L.~Pareschi, and M.~Zanella.
\newblock Opinion dynamics over complex networks: kinetic modelling and
  numerical methods.
\newblock {\em Kinet. Relat. Models}, 10(1):1--32, 2017.

\bibitem{barabasi1999SCIENCE}
A.-L. Barab\'asi and R.~Albert.
\newblock Emergence of scaling in random networks.
\newblock {\em Science}, 286(5439):509--512, 1999.

\bibitem{barabasi1999PHYSA}
A.-L. Barab\'{a}si, R.~Albert, and H.~Jeong.
\newblock Mean-field theory for scale-free random networks.
\newblock {\em Phys. A}, 272(1-2):73--187, 1999.

\bibitem{bisi2024Physd}
M.~Bisi and N.~Loy.
\newblock Kinetic models for systems of interacting agents with multiple
  microscopic states.
\newblock {\em Physica D: Nonlinear Phenomena}, 457:133967, 2024.

\bibitem{choi2019NJP}
J.~Choi and K.~I. Goh.
\newblock Majority-vote dynamics on multiplex networks with two layers.
\newblock {\em New Journal of Physics}, 21(3):035005, mar 2019.

\bibitem{clauset2009SIREV}
A.~Clauset, C.~R. Shalizi, and M.~E.~J. Newman.
\newblock Power-law distributions in empirical data.
\newblock {\em SIAM Rev.}, 51(4):661--703, 2009.

\bibitem{Krasno}
M.~A. M.~A. Krasnoselskii.
\newblock {\em Positive Solutions of Operator Equations}.
\newblock 1964.

\bibitem{lambiotte2007PRE}
R.~Lambiotte, M.~Ausloos, and J.~Holyst.
\newblock Majority model on a network with communities.
\newblock {\em Physical review. E, Statistical, nonlinear, and soft matter
  physics}, 75:030101, 04 2007.

\bibitem{LnMrTa}
N.~Loy, M.~Raviola, and A.~Tosin.
\newblock Opinion polarization in social networks.
\newblock {\em Philosophical Transactions of the Royal Society A},
  380(2224):20210158, 2022.

\bibitem{loy2020CMS}
N.~Loy and A.~Tosin.
\newblock Markov jump processes and collision-like models in the kinetic
  description of multi-agent systems.
\newblock {\em Commun. Math. Sci.}, 18(6):1539--1568, 2020.

\bibitem{loy2021KRM}
N.~Loy and A.~Tosin.
\newblock {B}oltzmann-type equations for multi-agent systems with label
  switching.
\newblock {\em Kinet. Relat. Models}, 2021.
\newblock doi:10.3934/krm.2021027.

\bibitem{loy2021MBE}
N.~Loy and A.~Tosin.
\newblock A viral load-based model for epidemic spread on spatial networks.
\newblock {\em Math. Biosci. Eng.}, 18(5):5635--5663, 2021.

\bibitem{MICHEL2005}
Philippe Michel, Stéphane Mischler, and Benoît Perthame.
\newblock General relative entropy inequality: an illustration on growth
  models.
\newblock {\em Journal de Mathématiques Pures et Appliquées},
  84(9):1235--1260, 2005.

\bibitem{newman2002PNAS}
M.~E.~J. Newman, D.~J. Watts, and S.~H. Strogatz.
\newblock Random graph models of social networks.
\newblock {\em Proc. Natl. Acad. Sci. USA}, 99(suppl 1):2566--2572, 2002.

\bibitem{nguyen2020SciRep}
V~Nguyen, G.~Xiao, X-J. Xu, Q.~Wu, and C-Y. Xia.
\newblock Dynamics of opinion formation under majority rules on complex social
  networks.
\newblock {\em Scientific Reports}, 10:456, 01 2020.

\bibitem{Noonan2021PRE}
J.~Noonan and R.~Lambiotte.
\newblock Dynamics of majority rule on hypergraphs.
\newblock {\em Phys. Rev. E}, 104:024316, Aug 2021.

\bibitem{peng2022}
K.~Peng and M.~A. Porter.
\newblock A majority-vote model on multiplex networks with community structure.
\newblock 2022.

\bibitem{benoit}
B.~Perthame.
\newblock {\em Transport Equations in Biology}.
\newblock Birkhäuser Basel, 2007.

\bibitem{slanina2003EPJB}
F.~Slanina and H.~Lavi\v{c}ka.
\newblock Analytical results for the {S}znajd model of opinion formation.
\newblock {\em Eur. Phys. J. B}, 35:279--288, 2003.

\bibitem{sznajd-weron2000IJMP}
K.~Sznajd-Weron and J.~Sznajd.
\newblock Opinion evolution in closed community.
\newblock {\em Internat. J. Modern Phys. C}, 11(6):1157--1165, 2000.

\bibitem{sznajd-weron2021PHYSA}
K.~Sznajd-Weron, J.~Sznajd, and T.~Weron.
\newblock A review on the {S}znajd model -- 20 years after.
\newblock {\em Phys. A}, 565(1):125537/1--12, 2021.

\bibitem{toscani2006CMS}
G.~Toscani.
\newblock Kinetic models of opinion formation.
\newblock {\em Commun. Math. Sci.}, 4(3):481--496, 2006.

\bibitem{toscani2018PRE}
G.~Toscani, A.~Tosin, and M.~Zanella.
\newblock Opinion modeling on social media and marketing aspects.
\newblock {\em Phys. Rev. E}, 98(2):022315/1--15, 2018.

\bibitem{watts1998NATURE}
D.~J. Watts and S.~H. Strogatz.
\newblock Collective dynamics of `small-world' networks.
\newblock {\em Nature}, 393:440--442, 1998.

\end{thebibliography}
\end{document}